\documentclass[sigconf]{acmart} 
\usepackage[nodisplayskipstretch]{setspace}
\usepackage{graphicx}
\usepackage{mathrsfs}
\usepackage{gensymb}
\usepackage{multirow}
\usepackage{enumitem}
\usepackage{array}
\usepackage{subfigure}
\usepackage{appendix}

\setlist{leftmargin=10pt}
\AtBeginDocument{%
  \providecommand\BibTeX{{%
    \normalfont B\kern-0.5em{\scshape i\kern-0.25em b}\kern-0.8em\TeX}}}
\allowdisplaybreaks

\copyrightyear{2024}
\acmYear{2024}
\setcopyright{rightsretained}
\acmConference[MOBIHOC '24]{The Twenty-fifth International Symposium on Theory, Algorithmic Foundations, and Protocol Design for Mobile Networks and Mobile Computing}{October 14--17, 2024}{Athens, Greece}
\acmBooktitle{The Twenty-fifth International Symposium on Theory, Algorithmic Foundations, and Protocol Design for Mobile Networks and Mobile Computing (MOBIHOC '24), October 14--17, 2024, Athens, Greece}\acmDOI{10.1145/3641512.3686376}
\acmISBN{979-8-4007-0521-2/24/10}

\usepackage{mathtools}
\usepackage{hyperref}
\usepackage{algorithmic}
\usepackage[ruled,linesnumbered]{algorithm2e}
\usepackage{bm}
\def\BibTeX{{\rm B\kern-.05em{\sc i\kern-.025em b}\kern-.08em
    T\kern-.1667em\lower.7ex\hbox{E}\kern-.125emX}}
\usepackage{pifont}
\usepackage{amsmath,amsfonts,amsthm}
\begin{document}
\setlength{\abovedisplayskip}{0pt plus 0pt minus 0pt}
\setlength{\belowdisplayskip}{0pt plus 0pt minus 0pt}
\setlength\abovedisplayshortskip{0pt plus 0pt minus 0pt}
\setlength\belowdisplayshortskip{0pt plus 0pt minus 0pt}
\title{Data Processing Efficiency Aware User Association and Resource Allocation in Blockchain Enabled Metaverse over Wireless Communications} 
\newtheorem{condition}{Condition}
\newtheorem{remark}{Remark}
\newtheorem{definitionx}{Definition}
\newtheorem{theoremx}{Theorem}
\newenvironment{talign}
 {\let\displaystyle\textstyle\align}
 {\endalign}
\author{Liangxin Qian}
\affiliation{%
  \institution{College of Computing and Data Science\\Nanyang Technological University, Singapore}
  \country{qian0080@e.ntu.edu.sg}
  }
\author{Jun Zhao}
\affiliation{%
  \institution{College of Computing and Data Science\\Nanyang Technological University, Singapore}
  \country{junzhao@ntu.edu.sg}
  }
 \thanks{Corresponding author: Jun Zhao\\Liangxin Qian is a PhD student supervised by Jun Zhao.}
\renewcommand{\shortauthors}{}
\begin{abstract}
In the rapidly evolving landscape of the Metaverse, enhanced by blockchain technology, the efficient processing of data has emerged as a critical challenge, especially in wireless communication systems. Addressing this need, our paper introduces the innovative concept of data processing efficiency (DPE), aiming to maximize processed bits per unit of resource consumption in blockchain-empowered Metaverse environments. To achieve this, we propose the \underline{D}PE-\underline{A}ware \underline{U}ser Association and \underline{R}esource Allocation (DAUR) algorithm, a tailored solution for these complex systems. The DAUR algorithm transforms the challenging task of optimizing the sum of DPE ratios into a solvable convex optimization problem. It uniquely alternates the optimization of key variables like user association, work offloading ratios, task-specific computing resource distribution, bandwidth allocation, user power usage ratios, and server computing resource allocation ratios. Our extensive numerical results demonstrate the DAUR algorithm's effectiveness in DPE.
\end{abstract}

\ccsdesc[500]{Networks}
\ccsdesc[500]{Networks~Network resources allocation}
\keywords{Blockchain, data processing efficiency, fractional programming, Metaverse, resource allocation, semidefinite relaxation.}

\maketitle
\thispagestyle{plain}
\pagestyle{plain}

\section{Introduction}
As blockchain technology and the Metaverse rapidly evolve, the efficient processing of non-fungible token (NFT) tasks is becoming crucial \cite{wang2021non}. NFTs, distinctive digital assets maintained on a blockchain, are gaining widespread popularity in numerous industries \cite{nadini2021mapping}. The growing interest in NFTs has created a significant demand for systems capable of effectively managing these tasks, particularly vital in the Metaverse where swift interactions and secure transactions are key.

Combining blockchain with the Metaverse introduces substantial computing demands, especially in managing NFTs \cite{christodoulou2022nfts}. These tasks require considerable processing power, often exceeding the capabilities of typical user devices. This scenario underscores the importance of developing intelligent strategies for user connectivity and resource distribution in wireless networks, ensuring the proficient handling of NFT tasks in this advanced digital landscape.

\textbf{Challenges and motivations.}
The evolved field of blockchain-enabled Metaverse, particularly the processing of NFT tasks, presents a unique set of challenges and motivations \cite{gadekallu2022blockchain}. One major challenge lies in the computational intensity of NFT operations. These tasks demand substantial resources, often exceeding what individual user devices can offer. Coupled with this is the need for real-time interactions within the Metaverse, where minimizing latency is critical to maintaining an immersive experience. Besides, the significant energy consumption in processing and managing NFTs on blockchain networks brings challenges of balancing energy efficiency with robust performance.

In response to these challenges, our research is motivated by the goal of maximizing data processing efficiency (DPE). Improving DPE in such environments is key to enhancing the efficiency of NFT task processing, leading to more sustainable and effective network operations. This study is also driven by the popularity of NFTs, highlighting the necessity to develop a supportive infrastructure that can accommodate their growth within the Metaverse. Achieving real-time interactions for NFT transactions and activities within the Metaverse is another critical motivation, essential for preserving the interactive essence of these digital realms.

Furthermore, this research focuses on optimizing computational and communication resources in wireless networks. This involves not only enhancing the performance of NFT tasks but also conserving energy, a balance that is increasingly important in our resource-conscious world. Therefore, we seek to navigate these challenges for the management and processing of NFTs in the blockchain-driven Metaverse, thereby making these technologies more efficient and integrated into digital interactions.

\textbf{Studied problem.}
Our research centers on a system with several users and servers, where the users delegate their NFT tasks to the servers. This offloading process is a strategic exercise in optimizing resource allocation and data processing to maximize DPE, a crucial metric in this context. DPE, representing the ratio of processed data bits to the sum of delay and energy consumption ($\frac{\textnormal{processed bits}}{\textnormal{delay} + \textnormal{energy}}$), offers a comprehensive evaluation of system performances.

We are exploring how to optimize DPE within this unique environment. Our goal is to devise a framework that not only boosts the efficiency of NFT task processing but also ensures an engaging and fluid user experience in the Metaverse, which are key features of the proposed framework. This involves addressing the intricacies of user-server connections and smart allocation of computational and communication resources in wireless networks.

\textbf{Main contributions.}
Our contributions are listed as follows:
\begin{itemize}
    \item[$\bullet$] Introduction of data processing efficiency (DPE): One contribution of this paper is the definition and exploration of the concept of DPE in blockchain-empowered Metaverse wireless communication systems. The study aims to achieve the highest possible DPE for each unit of resource consumed. This novel approach to efficiency measurement in the blockchain-Metaverse context sets a new benchmark for evaluating system performance.
    \item[$\bullet$] Development of the DAUR algorithm: We introduce the innovative \underline{D}PE \underline{A}ware \underline{U}ser association and \underline{R}esource allocation (DAUR) algorithm for blockchain-powered Metaverse wireless communications. This algorithm is a significant advancement as it simplifies the complex task of optimizing the sum of DPE ratios into a solvable convex optimization problem. A unique aspect of the DAUR algorithm is its approach to alternately optimize two sets of variables: $\{$user association, work offloading ratio, task-specific computing resource distribution$\}$ and $\{$bandwidth allocation ratio, user transmit power usage ratio, user computing resource usage ratio, server computing resource allocation ratio$\}$. By optimizing these sets together rather than separately, the DAUR algorithm achieves superior optimization results, enhancing the overall system efficiency.
    \item[$\bullet$] Effectiveness of the proposed DAUR algorithm is further underscored by numerical results. These results demonstrate the algorithm's success in significantly improving DPE within the studied systems. They shows the DAUR algorithm's practical utility and its potential to enhance the performance of blockchain-empowered Metaverse wireless communication systems.
\end{itemize}

The structure of this paper is outlined as follows: Section \ref{sec.related_work} presents relevant literature. Section \ref{sec.system_model} describes the system model, while the optimization problem is formulated in Section \ref{sec.optimization_prob}. The DAUR algorithm, our proposed solution for the optimization problem, is introduced in \mbox{Section \ref{sec.DAUR_algo}}, with a subsequent analysis of its complexity in Section \ref{sec.complexity_analysis}. Simulation results are presented in Section~\ref{sec.simulation_results}. Conclusion and future directions are given in Section \ref{sec.conclusion}.


\section{Related work}\label{sec.related_work}
In this section, we discuss the related work on efficiency metrics and resource allocation in Blockchain and Metaverse.
\subsection{Efficiency metrics}
In wireless communications, there are a few important metrics that help us gauge how well a network performs. Spectral efficiency looks at how effectively a network uses its available bandwidth~\cite{hu2014energy}. This is especially important when the bandwidth is limited, as it tells us how much data can be transmitted within a certain frequency range. Energy efficiency, on the other hand, measures how much data can be sent for a given amount of energy \cite{wang2021lifesaving,zhou2017near,liu2021towards}. This is crucial for devices like smartphones and IoT devices, which often have limited power sources \cite{6672036}. Cost efficiency then comes into play, assessing how much data can be transmitted cost-effectively \cite{li2018cloudshare}. This balance between performance and cost is key for maintaining an efficient yet affordable network. Lastly, throughput efficiency is all about how much data can be managed in a specific area, which is vital in densely populated areas where network traffic is high \cite{ju2013throughput}.

\subsubsection{Differences between data processing efficiency and other efficiency metrics}
DPE we studied is defined as the ratio of processed data bits to the sum of delay and energy consumption. Compared to traditional metrics, DPE provides a more comprehensive evaluation by incorporating both time and energy aspects into data processing. For instance, while spectral efficiency focuses on bandwidth utilization and energy efficiency on energy per bit transmitted, DPE integrates these aspects, emphasizing the balance between delay and energy in processing data. Unlike cost efficiency measuring economic aspects, or throughput efficiency assessing data volume per area, DPE directly ties the efficiency of data processing to tangible network performance factors – delay and energy.

\subsection{Resource allocation in Blockchain systems}
Resource allocation in blockchain environments is a critical area of study, aiming to optimize various aspects of network performance under the unique constraints and opportunities presented by blockchain technology~\cite{liu2021proof}. Guo \textit{et al}. \cite{guo2019adaptive} develop a blockchain-based mobile edge computing framework that enhances throughput and Quality of Service (QoS) by optimizing spectrum allocation, block size, and number of producing blocks, using deep reinforcement learning (DRL). Deng \textit{et al}. \cite{deng2022blockchain} tackle the challenge of decentralized model aggregation in blockchain-assisted federated learning (FL), proposing a novel framework that optimizes long-term average training data size and energy consumption, employing a Lyapunov technique for dynamic resource allocation. Feng \textit{et al}.'s study \cite{feng2020joint} focuses on minimizing energy consumption and delays in a blockchain-based mobile edge computing system, though it lacks explicit detail in bandwidth and transmit power allocation. Li \textit{et al}. \cite{9983804} propose a blockchain-based IoT resource monitoring and scheduling framework that securely manages and shares idle computing resources across the network for edge intelligence tasks, ensuring reliability and fairness. Xu \textit{et al}. \cite{xu2019healthchain} introduce the Healthchain, a blockchain-based scheme for preserving the privacy of large-scale health data in IoT environments, enabling encrypted data access control and secure, tamper-proof storage of IoT data and doctor diagnoses. Finally, Liu \textit{et al}. \cite{liu2019efficient} employ game theory in blockchain-based femtocell networks to maximize the utility of users, addressing power allocation challenges. 

\subsection{Resource allocation in Metaverse systems}
Resource allocation also plays a pivotal role in forging immersive experiences within the Metaverse, a fact underscored by numerous research efforts. Zhao \textit{et al}. \cite{10368052} focus on optimizing the utility-cost ratio for Metaverse applications over wireless networks, employing a novel fractional programming technique to enhance VR video quality through optimized communication and computation resources. Meanwhile, Chu \textit{et al}. \cite{chu2023metaslicing} introduce MetaSlicing, a framework that effectively manages and allocates diverse resources by grouping applications into clusters, using a semi-Markov decision process to maximize resource utilization and Quality-of-Service. This approach dramatically improves efficiency compared to traditional methods. On the other hand, Ng \textit{et al}. \cite{ng2022unified} tackle the unified resource allocation in a virtual education setting within the Metaverse, proposing a stochastic optimal resource allocation scheme to minimize service provider costs while adapting to the users' demand uncertainty.
\section{System Model and Parameter Description}\label{sec.system_model}
The system in Fig. \ref{fig.system_model} includes $N$ users and $M$ servers, with indices $n \in \mathcal{N}:= \{1,2,\cdots,N\}$ for users and $m \in\mathcal{M}:=  \{1,2,\cdots,M\}$ for servers. Users offload part of their workload to a server to maximize DPE. The workload is transmitted wirelessly, with servers allocating radio resources for reception. Servers then split their computing resources between processing user data and handling blockchain tasks, ensuring efficient management of both in a blockchain integrated Metaverse wireless communication system. 
\begin{figure}[htbp]
\vspace{-0.1cm}
\centering
\includegraphics[width=0.47\textwidth]{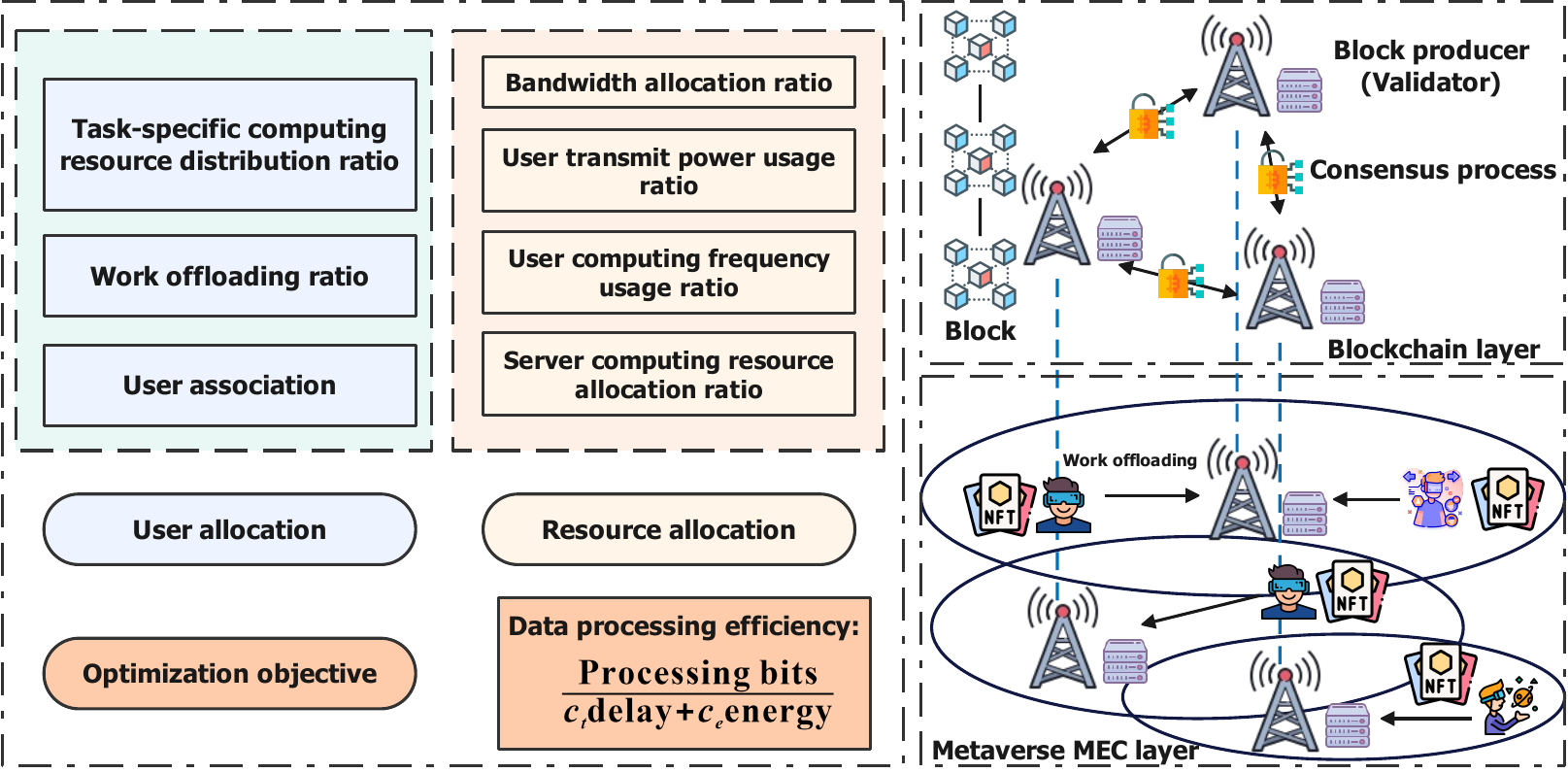}
\caption{Optimizing the DPE of a system consisting of $N$ VR users and $M$ Metaverse servers by the joint optimization of user association, offloading ratio, task-specific computing resource distribution ratio, and communication and computing resource allocation.}\vspace{-5pt}
\label{fig.system_model}
\end{figure}
\subsection{Parameters and variables in the system}
\subsubsection{User association}
We use $x_{n,m} \in \{0,1\}$ to denote the association between the user $n$ and the server $m$. Assume that every user is connected to one and only one server, i.e., $\sum_{m\in\mathcal{M}}x_{n,m}=1$.

\subsubsection{Partial work offloading}
Assume that $d_n$ is the total task of the user $n$. The work processed by user $n$ is $(1-\varphi_n)d_n$ and the work offloaded to server $m$ is $\varphi_n d_n$ with $\varphi_n \in [0,1]$. 

\subsubsection{Available computing resources of users}
The maximum computing resource of the user $n$ is $f_n$. We use $\psi_n \in [0,1]$ to denote the computing resource usage of the user $n$. Therefore, the computing resources used by the user $n$ to process \mbox{$(1-\varphi_n)d_n$} data is $\psi_nf_n$. 

\subsubsection{Wireless communication model}
Frequency-division multiple access (FDMA) is considered in this paper so that there is no interference in wireless communication between users and servers. Assume that the total available bandwidth of the server $m$ is $b_m$. We use $\phi_{n,m}b_m$ to represent the bandwidth that the server $m$ allocates to the user $n$, where $\phi_{n,m} \in [0,1]$. Similarly, $\rho_np_n$ is used to indicate the used transmit power of the user $n$ with $\rho_n \in [0,1]$ and $p_n$ is the maximum transmit power of the user $n$. According to the Shannon formula~\cite{shannon1948mathematical}, the transmission rate between user $n$ and server $m$ is 
\begin{talign}
    r_{n,m} = \phi_{n,m}b_m\log_2(1 + \frac{g_{n,m}\rho_np_n}{\sigma^2\phi_{n,m}b_m}),
\end{talign}
where $\phi_{n,m}b_m$ is the allocated bandwidth from server $m$ to user $n$, $g_{n,m}$ is the channel attenuation between them, and $\sigma^2$ is the noise power spectral density.

\subsubsection{Task-specific computing resources distribution of servers}\label{section_partiioning_computing_resources_of_servers}
The maximum computing resource of server $m$ is $f_m$. We use $\zeta_{n,m}$ to denote the allocated computing resource part from the chosen server $m$ to the user $n$ with $\zeta_{n,m} \in [0,1]$. We assume that the server allocates part of its computing resources $\gamma_{n,m}\zeta_{n,m}f_m$ to user $n$ for data processing and another portion $(1-\gamma_{n,m})\zeta_{n,m}f_m$ for processing blockchain tasks, where $\gamma_{n,m}\in (0,1)$ and $\zeta_{n,m}f_m$ is the allocated computing resources from the chosen server $m$ to the user $n$. Since the server $m$ needs to process data and also complete the blockchain task, it must reserve a portion of computing resources for these two parts (i.e., $\gamma_{n,m} \neq 0$ or $1$).

\subsection{Cost analysis}
In this section, we will analyze the user and server costs of the uplink. The connection $x_{n,m}$ between the user $n$ and the server $m$ can be decided based on the maximization of data processing efficiency, which will be detailed in Definition \ref{def_dpe}.
\subsubsection{User cost analysis}
The local data bits for processing are $(1-\varphi_n)d_n$ and the local processing delay is  
\begin{talign}
    T^{(up)}_n = \frac{(1-\varphi_n)d_n\eta_n}{\psi_nf_n},\label{eq_T_up}
\end{talign}
where $\eta_n$ (cycles/bit) is the number of CPU cycles per bit of the user $n$ and the superscript ``$up$'' means user processing. Then, we study the energy consumption of the user side. If the available computing resources of user $n$ is $\psi_nf_n$, the local processing energy consumption of the user $n$ is 
\begin{talign}
    E^{(up)}_n = \kappa_n(1-\varphi_n)d_n\eta_n(\psi_nf_n)^2,
\end{talign}
where $\kappa_n$ is the effective switched capacitance of the user $n$. After processing the local data bits, the user $n$ needs to transmit the rest $\varphi_n d_n$ data bits to the connected server $m$ over the wireless channel. The transmission delay of the user $n$ is 
\begin{talign}
    T^{(ut)}_{n,m} = \frac{x_{n,m}\varphi_nd_n}{r_{n,m}},
\end{talign} 
where the superscript ``$ut$'' represents user transmission. Assume that the available transmit power of the user $n$ is $\rho_np_n$, the energy consumption during the transmission of the user $n$ is 
\begin{talign}
    E^{(ut)}_{n,m} = \rho_np_nT^{(ut)}_{n,m} = \frac{x_{n,m}\rho_np_n\varphi_nd_n}{r_{n,m}}.
\end{talign}

\subsubsection{Server cost analysis}
After $\varphi_n d_n$ data bits are transmitted to the server $m$ from the user $n$, the server $m$ would first process those data bits. According to the discussion in section \ref{section_partiioning_computing_resources_of_servers}, the processing delay of the server $m$ for the offloaded work from the user $n$ is 
\begin{talign}
    T^{(sp)}_{n,m} = \frac{x_{n,m}\varphi_nd_n\eta_m}{\gamma_{n,m}\zeta_{n,m}f_m},
\end{talign}
and its related data processing energy consumption is 
\begin{talign}
    E^{(sp)}_{n,m} = \kappa_mx_{n,m}\varphi_nd_n\eta_m(\gamma_{n,m}\zeta_{n,m}f_m)^2,
\end{talign}
where the superscript ``$sp$'' stands for server processing and $\eta_m$ is the number of CPU cycles per bit of the server $m$. When $\varphi_n d_n$ data bits are processed by the server $m$, it will generate a blockchain block for them. Assume that in this phase, the server $m$ would process $\varphi_nd_n\omega_b$ bits, where $\omega_b$ is the data size changing ratio of servers mapping the offloaded data to the format that will be processed by the blockchain. Since $\gamma_{n,m}$ fraction of computing resources $\zeta_{n,m}f_m$ that allocated to the user $n$ is used in the data processing phase, there is still $1-\gamma_{n,m}$ fraction of that for block generation phase. Therefore, the delay of the server $m$ generating blockchain block for the user $n$ is 
\begin{talign}
    T^{(sg)}_{n,m} = \frac{x_{n,m}\varphi_nd_n\omega_b\eta_m}{(1-\gamma_{n,m})\zeta_{n,m}f_m},
\end{talign}
and its corresponding energy consumption is 
\begin{talign}
    E^{(sg)}_{n,m} = \kappa_mx_{n,m}\varphi_nd_n\eta_m \omega_b [(1-\gamma_{n,m})\zeta_{n,m}f_m]^2,
\end{talign}
where ``$sg$'' indicates server generation. Next, we discuss the block propagation and validation cost. Assume that there is only one hop among servers. We consider only the block propagation delay and the validation delay, as in \cite{feng2020joint}. According to the analysis in \cite{feng2020joint}, we can obtain the total block propagation delay in the data transactions during the consensus as
\begin{talign}
    T^{(bp)}_{n,m} = \frac{S_b}{R_m},
\end{talign}
where $R_m:=\text{min}_{m^\prime\in\mathcal{M}\setminus\{m\}}R_{m,m^\prime}$ and $R_{m,m^\prime}$ is the wired link transmission rate between the servers $m$ and $m^\prime$ and the superscript ``\textit{bp}'' means block propagation. Besides, based on \cite{feng2020joint}, the validation time during the block propagation is 
\begin{align}
    T^{(sv)}_{n,m} = \text{max}_{m^\prime\in\mathcal{M}\setminus\{m\}}\frac{\eta_v}{(1-\gamma_{n,m^\prime})\zeta_{n,m^\prime}f_m^\prime},
\end{align}
where $\eta_v$ is the number of CPU cycles required by server $m^\prime$ to verify the block and the superscript ``$sv$'' is server validation.

\section{Definition of Data Processing Efficiency and Studied Optimization Problem}\label{sec.optimization_prob}
Here we first introduce the concept of data processing efficiency before formulating the optimization problem.
\begin{definition}[Data Processing Efficiency]\label{def_dpe}
    Data processing efficiency (DPE) refers to the costs including delay and energy required to process data bits, i.e., DPE $:= \frac{\text{processing data bits}}{\text{delay + energy}}$. Specifically, in the two-level system (i.e., user and server), the cost of processing data at the user level includes the local data processing costs. The cost of processing data at the server level consists of both the data processing consumption and the wireless transmission consumption from the user level. For instance, if $\varphi_n d_n$ parameters are processed at the server $m$, the cost for processing these bits encompasses the delay and energy consumed in processing, as well as those involved in data transmission from the \mbox{user $n$}.
\end{definition}
Based on the concept of DPE, we obtain the cost required to process $(1-\varphi_n)d_n$ at the user $n$:
\begin{talign}
    cost^{(u)}_n = \omega_t T^{(up)}_n + \omega_e E^{(up)}_n, \label{eq.cost_u}
\end{talign}
where $\omega_t$ and $\omega_e$ denote the weight parameters of delay and energy consumption. Similarly, we get the cost needed to process $\varphi_n d_n$ at the server $m$:
\begin{talign}
    cost^{(s)}_{n,m} &= \omega_t (T^{(ut)}_{n,m} + T^{(sp)}_{n,m} + T^{(sg)}_{n,m} + T^{(bp)}_{n,m} + T^{(sv)}_{n,m}) \nonumber \\ 
    &+ \omega_e (E^{(ut)}_{n,m} + E^{(sp)}_{n,m} + E^{(sg)}_{n,m}). \label{eq.cost_s}
\end{talign}
Then, our optimization objective is to maximize the total DPE of all servers and users. Let $\bm{x}:=[x_{n,m}]|_{n\in\mathcal{N},m\in\mathcal{M}}$, $\bm{\varphi}:=[\varphi_n]|_{n\in\mathcal{N}}$, $\bm{\gamma}:=[\gamma_{n,m}]|_{n\in\mathcal{N},m\in\mathcal{M}}$, $\bm{\phi}:=[\phi_{n,m}]|_{n\in\mathcal{N},m\in\mathcal{M}}$, $\bm{\rho}:=[\rho_n]|_{n\in\mathcal{N}}$, $\bm{\zeta}:=[\zeta_{n,m}]|_{n\in\mathcal{N},m\in\mathcal{M}}$, and $\bm{\psi}:=[\psi_n]|_{n\in\mathcal{N}}$. We give the optimization problem $\mathbb{P}_{1}$ as follows:
\begin{subequations}\label{prob1}
\begin{align}
\mathbb{P}_{1}:&\!\!\!\!\max\limits_{\bm{x},\bm{\varphi},\bm{\gamma},\bm{\phi},\bm{\rho},\bm{\zeta},\bm{\psi}}  \!\!\sum\limits_{n \in \mathcal{N}} \frac{c_n(1-\varphi_n)d_n}{cost^{(u)}_n} +\!\! \sum\limits_{n \in \mathcal{N}}\sum\limits_{m \in \mathcal{M}}\frac{c_{n,m}x_{n,m}\varphi_nd_n}{cost^{(s)}_{n,m}}
\tag{\ref{prob1}}\\
\text{s.t.} \quad 
& x_{n,m} \in \{0,1\}, ~\forall n \in \mathcal{N}, ~\forall m \in \mathcal{M},\label{x_constr1} \\
& \sum\limits_{m \in \mathcal{M}} x_{n,m} = 1, ~\forall n \in \mathcal{N},\label{x_constr2} \\
& \varphi_n \in [0,1], ~\forall n \in \mathcal{N},\label{varphi_constr}\\
&\gamma_{n,m} \in (0,1), ~\forall n \in \mathcal{N}, ~\forall m \in \mathcal{M},\label{gamma_constr}\\
&\phi_{n,m} \in [0,1], ~\forall n \in \mathcal{N}, ~\forall m \in \mathcal{M},\label{phi_constr1}\\
&\sum\limits_{n\in \mathcal{N}} x_{n,m} \phi_{n,m} \leq 1, ~\forall m \in \mathcal{M},\label{phi_constr2}\\
&\zeta_{n,m} \in [0,1], ~\forall n \in \mathcal{N}, ~\forall m \in \mathcal{M},\label{zeta_constr1}\\
&\sum\limits_{n\in \mathcal{N}} x_{n,m} \zeta_{n,m} \leq 1, ~\forall m \in \mathcal{M},\label{zeta_constr2}\\
&\rho_n \in [0,1], ~\forall n \in \mathcal{N},\label{rho_constr}\\
&\psi_n \in [0,1], ~\forall n \in \mathcal{N},\label{psi_constr}
\end{align}
\end{subequations}
where $c_n$ is the DPE preference of the user $n$, $c_{n,m}$ is that between the server $m$ and the user $n$, and ``s.t.'' is short for ``subject to''. 
\section{Our Proposed DAUR Algorithm to Solve the Optimization Problem}\label{sec.DAUR_algo}
Problem $\mathbb{P}_{1}$ is a sum of ratios problem coupled by many complicated non-convex terms, which is generally NP-complete and difficult to solve directly. We will make Problem $\mathbb{P}_{1}$ a solvable convex problem within a series of transformations by our proposed \underline{D}PE-\underline{A}ware \underline{U}ser association and \underline{R}esource allocation (DAUR) algorithm. 
\begin{theorem}\label{theorem_solvep1}
    Problem $\mathbb{P}_{1}$ can be transformed into a solvable problem if we alternatively optimize $[\bm{x},\bm{\varphi},\bm{\gamma}]$ and $[\bm{\phi},\bm{\rho},\bm{\zeta},\bm{\psi}]$.
\end{theorem}
\begin{proof}
    \textbf{Theorem \ref{theorem_solvep1}} is proven by the following \textbf{Lemma \ref{lemma_p1top2}}, \textbf{Lemma \ref{lemma_p2top3}}, \textbf{Theorem \ref{theorem_p4toconcave}} in Section \ref{sec.AO_step1}, and \textbf{Theorem \ref{theorem_p7toconvex}} in Section \ref{sec.AO_step2}.
\end{proof}
\begin{lemma}\label{lemma_p1top2}
Define new auxiliary variables $\vartheta_n^{(u)}$, $\vartheta_{n,m}^{(s)}$, $T_n^{(u)}$, and $T_{n,m}^{(s)}$. Let $\bm{\vartheta^{(u)}}:=[\vartheta_n^{(u)}]|_{n\in\mathcal{N}}$, $\bm{\vartheta^{(s)}}:=[\vartheta_{n,m}^{(s)}]|_{n\in\mathcal{N},m\in\mathcal{M}}$, $\bm{T^{(u)}}:=[T_n^{(u)}]|_{n\in\mathcal{N}}$, $\bm{T^{(s)}}:=[T_{n,m}^{(s)}]|_{n\in\mathcal{N},m\in\mathcal{M}}$, $\bm{T}:=\{\bm{T^{(u)}},\bm{T^{(s)}}\}$, and $\bm{\vartheta}:=\{\bm{\vartheta^{(u)}},\bm{\vartheta^{(s)}}\}$. The sum of ratios Problem $\mathbb{P}_{1}$ can be transformed into a summation Problem $\mathbb{P}_{2}$:
\begin{subequations}\label{prob2}
\begin{align}
&\mathbb{P}_{2}:\max\limits_{\bm{x},\bm{\varphi},\bm{\gamma},\bm{\phi},\bm{\rho},\bm{\zeta},\bm{\psi}, \bm{\vartheta}, \bm{T}}  \sum\limits_{n \in \mathcal{N}} \vartheta_n^{(u)} + \sum\limits_{n \in \mathcal{N}}\sum\limits_{m \in \mathcal{M}}\vartheta_{n,m}^{(s)} \tag{\ref{prob2}}\\
&\text{s.t.} \quad (\text{\ref{x_constr1}})\text{-}(\text{\ref{psi_constr}}), \nonumber \\
& \omega_t T^{(u)}_n + \omega_e \kappa_n(1-\varphi_n)d_n\eta_n(\psi_nf_n)^2 - \frac{c_n(1-\varphi_n)d_n}{\vartheta_n^{(u)}}\leq 0, \label{vartheta_u_constr}\\
& \omega_t T^{(s)}_{n,m} - \frac{c_{n,m}x_{n,m}\varphi_nd_n}{\vartheta_{n,m}^{(s)}} + \omega_e \{\frac{x_{n,m}\rho_np_n\varphi_nd_n}{r_{n,m}} \nonumber \\ 
    &+ \kappa_mx_{n,m}\varphi_nd_n\eta_m(\gamma_{n,m}\zeta_{n,m}f_m)^2 \nonumber \\ 
    &+ \kappa_mx_{n,m}\varphi_nd_n\eta_m \omega_b [(1-\gamma_{n,m})\zeta_{n,m}f_m]^2\}  \leq 0, \label{vartheta_s_constr}\\
&T^{(up)}_n \leq T^{(u)}_n, \label{Tu_constr}\\
&T^{(ut)}_{n,m} + T^{(sp)}_{n,m} + T^{(sg)}_{n,m} + T^{(bp)}_{n,m} + T^{(sv)}_{n,m} \leq T^{(s)}_{n,m}. \label{Ts_constr}
\end{align}
\end{subequations}
\end{lemma}
\begin{proof}
    Refer to Appendix \ref{append_lemma_p1top2} in the full version paper~\cite{fullversion}.
\end{proof}
According to \textbf{Lemma \ref{lemma_p1top2}}, we can transform the sum of ratios Problem $\mathbb{P}_{1}$ to a summation Problem $\mathbb{P}_{2}$ by adding the extra auxiliary variables $\vartheta_n^{(u)}$, $\vartheta_{n,m}^{(s)}$, $T_n^{(u)}$, and $T_{n,m}^{(s)}$. Thanks to $\vartheta_n^{(u)}$ and $\vartheta_{n,m}^{(s)}$, we convert the sum of ratios of the objective function in \mbox{Problem $\mathbb{P}_{1}$} to the sum of two variables. Besides, we can transfer the troublesome terms about the delay of the objective function in \mbox{Problem $\mathbb{P}_{1}$} into the constraints (\ref{Tu_constr}) and (\ref{Ts_constr}) by introducing the variables $T_n^{(u)}$ and $T_{n,m}^{(s)}$. However, the constraints (\ref{vartheta_u_constr}) and (\ref{vartheta_s_constr}) are not convex and \mbox{Problem $\mathbb{P}_{2}$} is still hard to solve.

\begin{lemma}\label{lemma_p2top3}
Define new auxiliary variables $\alpha^{(u)}_n$ and $\alpha^{(s)}_{n,m}$. Let $\bm{\alpha^{(u)}}:=[\alpha^{(u)}_n]|_{n\in\mathcal{N}}$, $\bm{\alpha^{(s)}}:=[\alpha^{(s)}_{n,m}]|_{n\in\mathcal{N},m\in\mathcal{M}}$, and $\bm{\alpha}:=\{\bm{\alpha^{(u)}}$, $\bm{\alpha^{(s)}}\}$. After that, Problem $\mathbb{P}_{2}$ can be transformed into Problem $\mathbb{P}_{3}$:
\begin{subequations}\label{prob3}
\begin{align}
&\mathbb{P}_{3}:\max\limits_{\bm{x},\bm{\varphi},\bm{\gamma},\bm{\phi},\bm{\rho},\bm{\zeta},\bm{\psi}, \bm{\vartheta}, \bm{\alpha}, \bm{T}}  \nonumber \\
&\sum\limits_{n \in \mathcal{N}} \alpha^{(u)}_n [c_n(1 - \varphi_n)d_n - \vartheta_n^{(u)}cost^{(u)}_n] \nonumber \\
&+ \sum\limits_{n \in \mathcal{N}}\sum\limits_{m \in \mathcal{M}}\alpha^{(s)}_{n,m} (c_{n,m}x_{n,m}\varphi_n d_n - \vartheta_{n,m}^{(s)}cost^{(s)}_{n,m})\tag{\ref{prob3}}\\
&\text{s.t.} \quad (\text{\ref{x_constr1}})\text{-}(\text{\ref{psi_constr}}), (\text{\ref{Tu_constr}})\text{-}(\text{\ref{Ts_constr}}). \nonumber
\end{align}
\end{subequations}
At Karush–Kuhn–Tucker (KKT) points of Problem $\mathbb{P}_{3}$, we can obtain that 
\begin{talign}
    &\alpha^{(u)}_n = \frac{1}{cost^{(u)}_n},\label{eq_alpha_u}\\
    &\alpha^{(s)}_{n,m} = \frac{1}{cost^{(s)}_{n,m}},\label{eq_alpha_s}\\
    &\vartheta_n^{(u)} = \frac{c_n(1 - \varphi_n)d_n}{cost^{(u)}_n},\label{eq_vartheta_u}\\
 \text{and }   &\vartheta_{n,m}^{(s)} = \frac{c_{n,m}x_{n,m}\varphi_n d_n}{cost^{(s)}_{n,m}}.\label{eq_vartheta_s}
\end{talign}
\end{lemma}
\begin{proof}
    Refer to Appendix \ref{append_lemma_p2top3} in the full version paper~\cite{fullversion}.
\end{proof}

Based on \textbf{Lemma \ref{lemma_p2top3}}, we can split the ratio form of the objective function in Problem $\mathbb{P}_{1}$ and transform the non-convex constraints (\ref{vartheta_u_constr}) and (\ref{vartheta_s_constr}) into the objective function in Problem $\mathbb{P}_{3}$ by introducing new auxiliary variables $\alpha^{(u)}_n$ and $\alpha^{(s)}_{n,m}$. Besides, based on the analysis of the KKT conditions of Problem $\mathbb{P}_{3}$, we can obtain the relationships between auxiliary variables $[\alpha^{(u)}_n$, $\alpha^{(s)}_{n,m}$, $\vartheta_n^{(u)}$, $\vartheta_{n,m}^{(s)}]$ and original variables $[x_{n,m}$, $\varphi_n$, $\gamma_{n,m}$, $\phi_{n,m}$, $\rho_n$, $\zeta_{n,m}$, $\psi_n$, $T^{(u)}_n$, $T^{(s)}_{n,m}]$ as Equations (\ref{eq_alpha_u}), (\ref{eq_alpha_s}), (\ref{eq_vartheta_u}), (\ref{eq_vartheta_s}). Next, we consider alternative optimization to solve this complex problem. At the \mbox{$i$-th} iteration, we first fix $\bm{\alpha}^{(i-1)}$ and $\bm{\vartheta}^{(i-1)}$, and then optimize $\bm{x}^{(i)},\bm{\varphi}^{(i)},\bm{\gamma}^{(i)},\bm{\phi}^{(i)},\bm{\rho}^{(i)},\bm{\zeta}^{(i)},\bm{\psi}^{(i)}, \bm{T}^{(i)}$. We then update $\bm{\alpha}^{(i)}$ and $\bm{\vartheta}^{(i)}$ according to their results. Repeat the above optimization steps until the relative difference between the objective function values of Problem $\mathbb{P}_{3}$ in the $i$-th and \mbox{$(i-1)$-th} iterations is less than an acceptable threshold, and we get a stationary point for Problem $\mathbb{P}_{3}$. Next, we analyze how to optimize $\bm{x},\bm{\varphi},\bm{\gamma},\bm{\phi},\bm{\rho},\bm{\zeta},\bm{\psi},\bm{T}$ with the given $\bm{\vartheta}, \bm{\alpha}$.

\subsection{Solve \texorpdfstring{$\bm{\phi},\bm{\rho},\bm{\zeta},\bm{\psi}$}{}, and \texorpdfstring{$\bm{T}$}{} with fixed \texorpdfstring{$\bm{x}$}{}, \texorpdfstring{$\bm{\gamma}$}{}, and \texorpdfstring{$\bm{\varphi}$}{}}\label{sec.AO_step1}
If we first fix $\bm{x}$, $\bm{\gamma}$ and $\bm{\varphi}$, Problem $\mathbb{P}_{3}$ would be the following new Problem $\mathbb{P}_{4}$:
\begin{subequations}\label{prob4}
\begin{align}
\mathbb{P}_{4}:&\max\limits_{\bm{\phi},\bm{\rho},\bm{\zeta},\bm{\psi}, \bm{T}}  \sum\limits_{n \in \mathcal{N}} \alpha^{(u)}_n [c_n(1 - \varphi_n)d_n - \vartheta_n^{(u)}cost^{(u)}_n] \nonumber \\
&+ \!\!\sum\limits_{n \in \mathcal{N}}\sum\limits_{m \in \mathcal{M}}\alpha^{(s)}_{n,m} (c_{n,m}x_{n,m}\varphi_n d_n - \vartheta_{n,m}^{(s)}cost^{(s)}_{n,m}) \tag{\ref{prob4}}\\
&\text{s.t.} \quad (\text{\ref{phi_constr1}})\text{-}(\text{\ref{psi_constr}}), (\text{\ref{Tu_constr}})\text{-}(\text{\ref{Ts_constr}}). \nonumber
\end{align}
\end{subequations}
Note that $cost^{(s)}_{n,m}$ is still not convex. 
\begin{theorem}\label{theorem_p4toconcave}
    Problem $\mathbb{P}_{4}$ can be transformed into a solvable concave optimization problem by a fractional programming (FP) technique.
\end{theorem}
\begin{proof}
   \textbf{Theorem \ref{theorem_p4toconcave}} is proven by the following \mbox{\textbf{Lemma \ref{lemma_p4top5}}}.
\end{proof}

\begin{lemma}\label{lemma_p4top5}
Define a new auxiliary variable $\upsilon^{(s)}_{n,m}$, where $\upsilon^{(s)}_{n,m} = \frac{1}{2x_{n,m}\rho_np_n\varphi_nd_nr_{n,m}}$. Rewrite $cost^{(s)}_{n,m}$ as a new variable $\widetilde{cost}^{(s)}_{n,m}$ with $\upsilon^{(s)}_{n,m}$. Let $\bm{\upsilon^{(s)}} := [\upsilon^{(s)}_{n,m}|_{\forall n \in \mathcal{N},\forall m \in \mathcal{M}}]$. The Problem $\mathbb{P}_{4}$ can be transformed into the following Problem $\mathbb{P}_{5}$:
\begin{subequations}\label{prob5}
\begin{align}
\mathbb{P}_{5}:&\max\limits_{\bm{\phi},\bm{\rho},\bm{\zeta},\bm{\psi}, \bm{\upsilon^{(s)}},\bm{T}}  \sum\limits_{n \in \mathcal{N}} \alpha^{(u)}_n [c_n(1 - \varphi_n)d_n - \vartheta_n^{(u)}cost^{(u)}_n] \nonumber \\
&+ \!\!\sum\limits_{n \in \mathcal{N}}\sum\limits_{m \in \mathcal{M}}\alpha^{(s)}_{n,m} (c_{n,m}x_{n,m}\varphi_n d_n - \vartheta_{n,m}^{(s)}\widetilde{cost}^{(s)}_{n,m}) \tag{\ref{prob5}}\\
&\text{s.t.} \quad (\text{\ref{phi_constr1}})\text{-}(\text{\ref{psi_constr}}),(\text{\ref{Tu_constr}})\text{-}(\text{\ref{Ts_constr}}). \nonumber
\end{align}
\end{subequations}
If we alternatively optimize $\bm{\upsilon^{(s)}}$ and $\bm{\phi},\bm{\rho},\bm{\zeta},\bm{\psi},\bm{T}$, Problem $\mathbb{P}_{5}$ would be a concave problem. Besides, with the local optimum $\bm{\upsilon}^{\bm{(s)}(\star)}$, we can find $\bm{\phi}^{(\star)},\bm{\rho}^{(\star)},\bm{\zeta}^{(\star)},\bm{\psi}^{(\star)}, \bm{T}^{(\star)}$, which is a stationary point of Problem $\mathbb{P}_{5}$.
\end{lemma}
\begin{proof}
    Refer to Appendix \ref{append_lemma_p4top5} in the full version paper~\cite{fullversion}.
\end{proof}
Thanks to \textbf{Lemma \ref{lemma_p4top5}}, we can transform Problem $\mathbb{P}_{4}$ into a solvable concave problem Problem $\mathbb{P}_{5}$ with the help of $\upsilon^{(s)}_{n,m}$. During the $i$-th iteration, we initially hold $\bm{\upsilon}^{\bm{(s)}(i-1)}$ constant and focus on optimizing $\bm{\phi}^{(i)}, \bm{\rho}^{(i)}, \bm{\zeta}^{(i)}, \bm{\psi}^{(i)}, \bm{T}^{(i)}$. Once these values are determined, we update $\bm{\upsilon}^{\bm{(s)}(i)}$ based on the obtained results. This optimization cycle is repeated until the difference in the objective function value of Problem $\mathbb{P}_{5}$ between the $i$-th and $(i-1)$-th iterations falls in a predefined threshold. Reaching this point signifies a solution for Problem $\mathbb{P}_{5}$, and consequently, for Problem $\mathbb{P}_{4}$. Next, we analyze how to optimize $\bm{x}$, $\bm{\gamma}$, $\bm{\varphi}$, and $\bm{T}$ with fixed $\bm{\phi},\bm{\rho},\bm{\zeta}$, and $\bm{\psi}$.

\subsection{Solve \texorpdfstring{$\bm{x}$}{}, \texorpdfstring{$\bm{\gamma}$}{}, \texorpdfstring{$\bm{\varphi}$}{}, and  \texorpdfstring{$\bm{T}$}{} with fixed \texorpdfstring{$\bm{\phi},\bm{\rho},\bm{\zeta}$}{}, and \texorpdfstring{$\bm{\psi}$}{}}\label{sec.AO_step2}
If $\bm{\phi},\bm{\rho},\bm{\zeta}$, and $\bm{\psi}$ are given, Problem $\mathbb{P}_{3}$ would be the following Problem:
\begin{subequations}\label{prob6}
\begin{align}
\mathbb{P}_{6}:&\max\limits_{\bm{x},\bm{\varphi},\bm{\gamma},\bm{T}}  \sum\limits_{n \in \mathcal{N}} \alpha^{(u)}_n [c_n(1 - \varphi_n)d_n - \vartheta_n^{(u)}cost^{(u)}_n] \nonumber \\
&+ \!\!\sum\limits_{n \in \mathcal{N}}\sum\limits_{m \in \mathcal{M}}\alpha^{(s)}_{n,m} (c_{n,m}x_{n,m}\varphi_n d_n - \vartheta_{n,m}^{(s)}cost^{(s)}_{n,m}) \tag{\ref{prob6}}\\
&\text{s.t.} \quad (\text{\ref{x_constr1}})\text{-}(\text{\ref{gamma_constr}}),~(\text{\ref{phi_constr2}}),~(\text{\ref{zeta_constr2}}),~(\text{\ref{Tu_constr}})\text{-}(\text{\ref{Ts_constr}}). \nonumber
\end{align}
\end{subequations}
Since $x_{n,m}$ is a binary discrete variable and other variables are continuous, Problem $\mathbb{P}_{6}$ is a mixed-integer nonlinear programming problem. We transform the constraint (\text{\ref{x_constr1}}) to $x_{n,m}(x_{n,m}-1)=0$. Therefore, Problem $\mathbb{P}_{6}$ can be further rewritten as:
\begin{subequations}\label{prob7}
\begin{align}
\mathbb{P}_{7}:&\max\limits_{\bm{x},\bm{\varphi},\bm{\gamma},\bm{T}}  \sum\limits_{n \in \mathcal{N}} \alpha^{(u)}_n [c_n(1 - \varphi_n)d_n - \vartheta_n^{(u)}cost^{(u)}_n] \nonumber \\
&+ \!\!\sum\limits_{n \in \mathcal{N}}\sum\limits_{m \in \mathcal{M}}\alpha^{(s)}_{n,m} (c_{n,m}x_{n,m}\varphi_n d_n - \vartheta_{n,m}^{(s)}cost^{(s)}_{n,m}) \tag{\ref{prob7}}\\
&\text{s.t.} \quad x_{n,m}(x_{n,m}-1)=0, \forall n \in \mathcal{N}, \forall m \in \mathcal{M}, \label{x_constr1_new}\\
&\quad\quad(\text{\ref{x_constr2}})\text{-}(\text{\ref{gamma_constr}}),~(\text{\ref{phi_constr2}}),~(\text{\ref{zeta_constr2}}),~(\text{\ref{Tu_constr}})\text{-}(\text{\ref{Ts_constr}}). \nonumber
\end{align}
\end{subequations}
\begin{theorem}\label{theorem_p7toconvex}
    Problem $\mathbb{P}_{7}$ can be transformed into a solvable convex optimization problem.
\end{theorem}
\begin{proof}
    \textbf{Theorem \ref{theorem_p7toconvex}} is proven by Lemmas \ref{lemma_gamma},
    \ref{lemma_p8top9}, and  \ref{lemma_p9top10} below.
\end{proof}

\begin{lemma}\label{lemma_gamma}
In Problem $\mathbb{P}_{7}$, if focusing on $cost_{n,m}^{(s)}$, $T^{(sp)}_{n,m}$, $T^{(sg)}_{n,m}$ and setting $\omega_b = 1$, we can obtain the optimum solution of $\gamma_{n,m}$ as $\gamma_{n,m}^\star = \frac{1}{2}$.
\end{lemma}
\begin{proof}
    Refer to Appendix \ref{append_lemma_gamma} in the full version paper~\cite{fullversion}.
\end{proof}
With \textbf{Lemma \ref{lemma_gamma}}, we obtain the optimum solution of $\gamma_{n,m}$ as $\gamma_{n,m}^\star = \frac{1}{2}$. By substituting the value of $\gamma_{n,m}^\star$, the objective function of Problem $\mathbb{P}_{7}$ becomes:
\begin{talign}
    &\sum_{n \in \mathcal{N}} \alpha^{(u)}_n [c_n(1 - \varphi_n)d_n - \vartheta_n^{(u)}cost^{(u)}_n] \nonumber \\
    &+ \sum_{n \in \mathcal{N}}\sum_{m \in \mathcal{M}}\alpha^{(s)}_{n,m} (c_{n,m}x_{n,m}\varphi_n d_n - \vartheta_{n,m}^{(s)}cost^{(s)}_{n,m})\nonumber \\
    &=\!\!\sum_{n \in \mathcal{N}}\{\alpha^{(u)}_n c_n d_n \!\!-\!\! \alpha^{(u)}_n \vartheta_n^{(u)} \omega_t T_n^{(u)} \!\!-\!\!\alpha^{(u)}_n \vartheta_n^{(u)} \omega_e \kappa_n d_n \eta_n f_n^2 \psi_n^2 \nonumber \\
    &- (\alpha^{(u)}_n c_n d_n - \alpha^{(u)}_n \vartheta_n^{(u)} \omega_e \kappa_n d_n \eta_n f_n^2 \psi_n^2) \varphi_n\} \nonumber \\
    &+ \sum_{n \in \mathcal{N}}\sum_{m \in \mathcal{M}}\{\alpha^{(s)}_{n,m}c_{n,m}d_n x_{n,m}\varphi_n - \alpha^{(s)}_{n,m} \vartheta_{n,m}^{(s)} \omega_t T^{(s)}_{n,m} \nonumber \\
    &- \alpha^{(s)}_{n,m} \vartheta_{n,m}^{(s)} \omega_e \frac{\rho_n p_n d_n}{r_{n,m}}x_{n,m}\varphi_n \nonumber \\
    &- \frac{1}{2}\alpha^{(s)}_{n,m} \vartheta_{n,m}^{(s)} \omega_e \kappa_m d_n \eta_m\zeta_{n,m}^2 f_m^2 x_{n,m}\varphi_n\}.
\end{talign}
For brevity, let $A_n:=\alpha^{(u)}_n c_n d_n - \alpha^{(u)}_n \vartheta_n^{(u)} \omega_e \kappa_n d_n \eta_n f_n^2 \psi_n^2$, $B_{n,m}:=\alpha^{(s)}_{n,m} \vartheta_{n,m}^{(s)} \omega_e \frac{\rho_n p_n d_n}{r_{n,m}} + \frac{1}{2}\alpha^{(s)}_{n,m} \vartheta_{n,m}^{(s)} \omega_e \kappa_m d_n \eta_m\zeta_{n,m}^2 f_m^2 - \alpha^{(s)}_{n,m}c_{n,m}d_n$, and $C:=\sum_{n \in \mathcal{N}}(-\alpha^{(u)}_n c_n d_n + \alpha^{(u)}_n \vartheta_n^{(u)} \omega_e \kappa_n d_n \eta_n f_n^2 \psi_n^2)$. Let $\bm{A}:=[A_n]|_{n \in \mathcal{N}}$ and $\bm{B}:=[B_{n,m}]|_{n \in \mathcal{N},m \in \mathcal{M}}$. Besides, the ``max'' problem in Problem $\mathbb{P}_{7}$ can also be rewritten as a ``min'' problem:
\begin{subequations}\label{prob8}
\begin{align}
\mathbb{P}_{8}:&\min\limits_{\bm{x},\bm{\varphi},\bm{T}}  C+\sum\limits_{n \in \mathcal{N}} 
\alpha^{(u)}_n \vartheta_n^{(u)} \omega_t T_n^{(u)} + A_n \varphi_n \nonumber \\
&+ \sum_{n \in \mathcal{N}}\sum_{m \in \mathcal{M}}\alpha^{(s)}_{n,m} \vartheta_{n,m}^{(s)} \omega_t T^{(s)}_{n,m} + B_{n,m} x_{n,m} \varphi_n \tag{\ref{prob8}}\\
&\text{s.t.} \quad (\text{\ref{x_constr1_new}}), (\text{\ref{x_constr2}})\text{-}(\text{\ref{varphi_constr}}),
(\text{\ref{phi_constr2}}), (\text{\ref{zeta_constr2}}),
(\text{\ref{Tu_constr}})\text{-}(\text{\ref{Ts_constr}}). \nonumber
\end{align}
\end{subequations}
Since there is the quadratic term $x_{n,m} \varphi_n$ and the quadratic constraint (\text{\ref{x_constr1_new}}), Problem $\mathbb{P}_{8}$ is a quadratically constrained quadratic programming (QCQP) problem. Next, we will explain how to transform Problem $\mathbb{P}_{8}$ into a standard QCQP form problem. We first use a new matrix $\bm{Q}:=(\bm{\varphi}^\intercal,\bm{x_1}^\intercal,\cdots,\bm{x_M}^\intercal)^\intercal$ to combine $\bm{x}$ and $\bm{\varphi}$, where $\bm{\varphi} = (\varphi_1,\cdots,\varphi_N)^\intercal$ and $\bm{x_m} = (x_{1,m},\cdots,x_{N,m})$. Here, we define some auxiliary matrices and vectors to facilitate our transformation. Let $e_i:=(0,\cdots,1_{i\text{-th}},\cdots,0)_{NM+N\times1}^\intercal$, $\bm{e}_{i,j}:=(e_i,\cdots,e_j)^\intercal$, $e_{\bar{i}}:=(0,\cdots,1_{i\text{-th}},\cdots,1_{(i+N)\text{-th}},\cdots,1_{[i+N(M-1)]\text{-th}},\cdots,0)^\intercal$, $\bm{e}_{\bar{i},\bar{j}}$$:=$$(e_{\bar{i}}$, $\cdots$, $e_{\bar{j}})^\intercal$, $e_{i\rightarrow j}:=$$(0,\cdots,1_{i\text{-th}}$, $1$, $\cdots$, $1_{j\text{-th}}$, $0$, $\cdots,0)^\intercal$, where $i<j$, $\bm{I}_{NM+N\times N}:=(\bm{I}_N,\bm{0}_{N\times NM})^\intercal$, and $\bm{I}_{N\rightarrow NM}:=(\bm{I}_N,\cdots,\bm{I}_N)_{N\times NM}$. Making $\sum_{n \in \mathcal{N}} 
\alpha^{(u)}_n \vartheta_n^{(u)} \omega_t T_n^{(u)}$ and $\sum_{n \in \mathcal{N}}\sum_{m \in \mathcal{M}}\alpha^{(s)}_{n,m} \vartheta_{n,m}^{(s)} \omega_t T^{(s)}_{n,m}$ succincter, we introduce two variables $T^{(u)}$ and $T^{(s)}$. Then, the constraints (\text{\ref{Tu_constr}}) and (\text{\ref{Ts_constr}}) will be
\begin{talign}
&\sum_{n \in \mathcal{N}} 
\alpha^{(u)}_n \vartheta_n^{(u)} \omega_t T_n^{(u)} \leq T^{(u)}, \label{Tu_constr1}\\
&\sum_{n \in \mathcal{N}}\sum_{m \in \mathcal{M}}\alpha^{(s)}_{n,m} \vartheta_{n,m}^{(s)} \omega_t T^{(s)}_{n,m} \leq T^{(s)},\label{Ts_constr1}
\end{talign}
respectively.
\begin{lemma}\label{lemma_p8top9}
There're matrices $\bm{P}_0$, $\bm{W}_0$, $\bm{P}_2^{(T_u)}$, $\bm{P}_0^{(T_s)}$, and terms $P_1^{(T_u)}$ and $P_1^{(T_s)}$ that can transform Problem $\mathbb{P}_{8}$ into the standard QCQP Problem $\mathbb{P}_{9}$:
\begin{subequations}\label{prob9}
\begin{align}
\mathbb{P}_{9}:&\min\limits_{\bm{Q},T^{(u)},T^{(s)}} \bm{Q}^\intercal \bm{P}_0 \bm{Q} + \bm{W}_0^\intercal \bm{Q} + T^{(u)} + T^{(s)} + C \tag{\ref{prob9}}\\
\text{s.t.} \quad 
&\text{diag}(\boldsymbol{e}_{N+1,NM+N}^\intercal\boldsymbol{Q})(\text{diag}(\boldsymbol{e}_{N+1,NM+N}^\intercal\boldsymbol{Q})-\bm{I})=\bm{0}, \label{x_constr1_qcqp}\\       
& \text{diag}(\boldsymbol{e}_{\overline{1},\overline{M}}^\intercal\boldsymbol{e}_{N+1,NM+N}^\intercal\boldsymbol{Q})=\bm{I}, \label{x_constr2_qcqp}\\
& \text{diag}(\boldsymbol{e}_{1,N}^\intercal\boldsymbol{Q})\preceq\bm{I}, \label{varphi_constr1_qcqp}\\
& \text{diag}(\boldsymbol{e}_{1,N}^\intercal\boldsymbol{Q})\succeq\bm{0}, \label{varphi_constr2_qcqp}\\
&\boldsymbol{\phi}^\intercal\boldsymbol{e}_{N+1,NM+N}^\intercal\boldsymbol{Q}-1 \leq 0, \label{x_phi_constr_qcqp}\\
&\boldsymbol{\zeta}^\intercal\boldsymbol{e}_{N+1,NM+N}^\intercal\boldsymbol{Q}-1 \leq 0, \label{x_zeta_constr_qcqp}\\
& {\bm{P}_2^{(T_u)}}^\intercal \bm{Q} + P_1^{(T_u)} \leq T^{(u)}, \label{Tu_constr_qcqp}\\
& \bm{Q}^\intercal \bm{P}_0^{(T_s)} \bm{Q} + P_1^{(T_s)} \leq T^{(s)}. \label{Ts_constr_qcqp}
\end{align}
\end{subequations}
\end{lemma}
\begin{proof}
    Refer to Appendix \ref{append_lemma_p8top9} in the full version paper~\cite{fullversion}.
\end{proof}

Using \textbf{Lemma \ref{lemma_p8top9}}, we can transform Problem $\mathbb{P}_{8}$ into the standard QCQP form Problem $\mathbb{P}_{9}$. However, Problem $\mathbb{P}_{9}$ is still non-convex. Then, we will use the semidefinite programming (SDP) method to transform this QCQP problem into an SDR problem. We introduce a new variable $\bm{S}:=(\bm{Q}^\intercal,1)^\intercal(\bm{Q}^\intercal,1)$. 

\begin{lemma}\label{lemma_p9top10}
There exist matrices $\bm{P}_1$, $\bm{P}_2$, $\bm{P}_3$, $\bm{P}_4$, $\bm{P}_5$, $\bm{P}_6$, $\bm{P}_7$, and $\bm{P}_8$ that can convert the QCQP Problem $\mathbb{P}_{9}$ into the SDR Problem $\mathbb{P}_{10}$:
\begin{subequations}\label{prob10}
\begin{align}
\mathbb{P}_{10}: &\min\limits_{\bm{S},T^{(u)},T^{(s)}}\quad  \text{Tr}(\bm{P}_1 \bm{S})\tag{\ref{prob10}}\\
\text{s.t.} \quad & \text{Tr}(\bm{P}_2 \bm{S})=0, \label{x_constr1_sdr}\\       & \text{Tr}(\bm{P}_3 \bm{S})=0, \label{x_constr2_sdr}\\
         & \text{Tr}(\bm{P}_4 \bm{S})\leq0, \label{varphi_constr_sdr}\\
         & \text{Tr}(\bm{P}_5 \bm{S})\leq0, \label{x_phi_constr_sdr}\\
         & \text{Tr}(\bm{P}_6 \bm{S})\leq0, \label{x_zeta_constr_sdr}\\
         & \text{Tr}(\bm{P}_7 \bm{S})\leq T^{(u)}, \label{Tu_constr_sdr}\\
         & \text{Tr}(\bm{P}_8 \bm{S})\leq T^{(s)}, \label{Ts_constr_sdr}\\
         & \bm{S}\succeq0, \label{S_constr_sdr}
\end{align}
\end{subequations}
where $Tr(\cdot)$ means the trace of a matrix.
\end{lemma}
\begin{proof}
    Refer to Appendix \ref{append_lemma_p9top10} in the full version paper~\cite{fullversion}.
\end{proof}

By \textbf{Lemma \ref{lemma_p9top10}}, Problem $\mathbb{P}_{10}$ is finally transformed into a solvable SDR Problem $\mathbb{P}_{10}$. Standard convex solvers can efficiently solve the SDR Problem $\mathbb{P}_{10}$ in polynomial time, providing a continuous version of $\bm{Q}$. However, this version often only serves as the lower bound for the ideal solution and may not satisfy the $\text{rank}(\mathbf{S})=1$ constraint. To rectify this, we apply rounding techniques. The final $NM$ components of $\bm{Q}$, represented by $x_{n,m}$ for every $n \in \mathcal{N}$ and $m \in \mathcal{M}$, reflect the partial connection of users to servers. If the sum $\sum_{m \in \mathcal{M}} x_{n,m}$ exceeds 1 for any user, we normalize $x_{n,m}$ by dividing it by the absolute sum. The Hungarian algorithm \cite{kuhn1955hungarian}, augmented with zero vectors, is used to identify the optimal matching, denoted as $\mathcal{X}$. Within this matching, we set $x_{n,m}$ to 1 if nodes $n$ and $m$ are paired, and 0 otherwise, labeling this as $\bm{x}^\star$. We set the results of $\bm{\varphi}$ in $\bm{Q}$ as $\bm{\varphi}^\star$.
\begin{figure}[htbp]
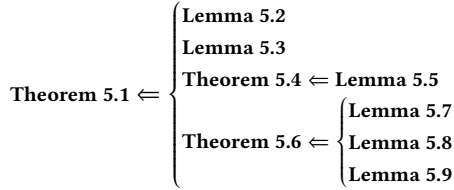

\begin{talign}
    \small\textnormal{\textbf{Theorem \ref{theorem_solvep1}}} \Leftarrow 
    \begin{cases}
        \textnormal{\textbf{Lemma \ref{lemma_p1top2}}}\\
        \textnormal{\textbf{Lemma \ref{lemma_p2top3}}}\\
        \textnormal{\textbf{Theorem \ref{theorem_p4toconcave}}} \Leftarrow \textnormal{\textbf{Lemma \ref{lemma_p4top5}}}\\
        \textnormal{\textbf{Theorem \ref{theorem_p7toconvex}}} \Leftarrow 
        \begin{cases}
            \textnormal{\textbf{Lemma \ref{lemma_gamma}}}\\
            \textnormal{\textbf{Lemma \ref{lemma_p8top9}}}\\
            \textnormal{\textbf{Lemma \ref{lemma_p9top10}}} \nonumber
        \end{cases}
    \end{cases}
\end{talign}
    \caption{Relationships between Theorems and Lemmas, where ``$\Leftarrow$'' means imply.}
    \label{fig:relat_theorem_lemma}
\end{figure}

\begin{algorithm}[htbp]
\caption{Proposed DAUR algorithm.}
\label{algorithm}

Initialize $i \leftarrow -1$ and for all $n \in \mathcal{N}, m \in \mathcal{M}$: $\bm{x}^{(0)} = (\bm{e_1},\cdots,\bm{e_M})^\intercal$, $\varphi_n^{(0)}=0.5$, $\phi_{n,m}^{(0)}=\frac{1}{N}$, $\rho_n^{(0)}=1$, $\zeta_{n,m}^{(0)}=\frac{1}{N}$, $\psi_n^{(0)}=1$, $\gamma_{n,m} = 0.5$;

Calculate $\bm{\alpha}^{(0)}, \bm{\vartheta}^{(0)}$ with $\bm{x}^{(0)} = (\bm{e_1},\cdots,\bm{e_M})^\intercal$, $\varphi_n^{(0)}=0.5$, $\phi_{n,m}^{(0)}=\frac{1}{N}$, $\rho_n^{(0)}=1$, $\zeta_{n,m}^{(0)}=\frac{1}{N}$, $\psi_n^{(0)}=1$, $\gamma_{n,m} = 0.5$;

\Repeat{$\frac{V_{\mathbb{P}_{3}}(\bm{x}^{(i+1)}, \bm{\varphi}^{(i+1)}, \bm{\phi}^{(i+1)}, \bm{\rho}^{(i+1)}, \bm{\zeta}^{(i+1)}, \bm{\psi}^{(i+1)})}{V_{\mathbb{P}_{3}}(\bm{x}^{(i)}, \bm{\varphi}^{(i)}, \bm{\phi}^{(i)}, \bm{\rho}^{(i)}, \bm{\zeta}^{(i)}, \bm{\psi}^{(i)})}- 1 \leq \epsilon_3$, where $\epsilon_3$ is a small positive number}{
Let $i\leftarrow i+1$;

Initialize $j = -1$; 

Calculate $\bm{\upsilon}^{(i,0)}$ with $\bm{x}^{(i)}, \bm{\varphi}^{(i)}, \bm{\phi}^{(i)}, \bm{\rho}^{(i)}, \bm{\zeta}^{(i)}, \bm{\psi}^{(i)}$;

Set $[\bm{\phi}^{(i,0)}, \bm{\rho}^{(i,0)}, \bm{\zeta}^{(i,0)}, \bm{\psi}^{(i,0)}, \bm{T}^{(i,0)}]\leftarrow$ $[\bm{\phi}^{(i)}, \bm{\rho}^{(i)}, \bm{\zeta}^{(i)}, \bm{\psi}^{(i)}, \bm{T}^{(i)}]$;

\Repeat{$\frac{V_{\mathbb{P}_{5}}(\bm{\phi}^{(i,j+1)}, \bm{\rho}^{(i,j+1)}, \bm{\zeta}^{(i,j+1)}, \bm{\psi}^{(i,j+1)})}{V_{\mathbb{P}_{5}}(\bm{\phi}^{(i,j)}, \bm{\rho}^{(i,j)}, \bm{\zeta}^{(i,j)}, \bm{\psi}^{(i,j)})}- 1 \leq \epsilon_1$, where $\epsilon_1$ is a small positive number}{
Let $j\leftarrow j+1$;

Obtain $[\bm{\phi}^{(i,j+1)}, \bm{\rho}^{(i,j+1)}, \bm{\zeta}^{(i,j+1)}, \bm{\psi}^{(i,j+1)}, \bm{T}^{(i,j+1)}]$ by solving Problem $\mathbb{P}_{5}$ with $\bm{\upsilon}^{(i,j)}$;

Update $\bm{\upsilon}^{(i,j+1)}$ with $[\bm{\phi}^{(i,j+1)}, \bm{\rho}^{(i,j+1)}, \bm{\zeta}^{(i,j+1)}, \bm{\psi}^{(i,j+1)}, \bm{T}^{(i,j+1)}]$;
}
Return $[\bm{\phi}^{(i,j+1)}, \bm{\rho}^{(i,j+1)}, \bm{\zeta}^{(i,j+1)}, \bm{\psi}^{(i,j+1)}]$ as a solution to Problem $\mathbb{P}_{5}$;

Set $[\bm{\phi}^{(i+1)}, \bm{\rho}^{(i+1)}, \bm{\zeta}^{(i+1)}, \bm{\psi}^{(i+1)}]$ $\leftarrow$ $[\bm{\phi}^{(i,j+1)}, \bm{\rho}^{(i,j+1)}, \bm{\zeta}^{(i,j+1)}, \bm{\psi}^{(i,j+1)}]$;

Initialize $j = -1$;

Set $[\bm{x}^{(i,0)}, \bm{\varphi}^{(i,0)}] \leftarrow [\bm{x}^{(i)}, \bm{\varphi}^{(i)}]$;

Initialize $[\bm{T}^{(i,0)}, \bm{A}^{(i,0)}, \bm{B}^{(i,0)}, C^{(i,0)}, \bm{P}_k^{(i,0)}]\leftarrow$ $[\bm{T}^{(i)}, \bm{A}^{(i)}, \bm{B}^{(i)}, C^{(i)}, \bm{P}_k^{(i)}]$, where $k \in \{1,2,\cdots,8\}$;


\Repeat{$\frac{V_{\mathbb{P}_{10}}(\bm{x}^{(i,j+1)}, \bm{\varphi}^{(i,j+1)}}{V_{\mathbb{P}_{10}}(\bm{x}^{(i,j)}, \bm{\varphi}^{(i,j)})}- 1 \leq \epsilon_2$, where $\epsilon_2$ is a small positive number}{
Let $j\leftarrow j+1$;

Obtain $[\bm{x}^{(i,j+1)}, \bm{\varphi}^{(i,j+1)}]$ of continuous values by solving Problem $\mathbb{P}_{10}$;

Update $[\bm{T}^{(i,j+1)}, \bm{A}^{(i,j+1)}, \bm{B}^{(i,j+1)}, C^{(i,j+1)}, \bm{P}_k^{(i,j+1)}$ with $[\bm{x}^{(i,j+1)}, \bm{\varphi}^{(i,j+1)}]$;
}

Return $[\bm{x}^{(i,j+1)}, \bm{\varphi}^{(i,j+1)}]$ as a solution to the SDR Problem $\mathbb{P}_{10}$;

Use the Hungarian algorithm augmented with zero vectors to obtain integer association results as $\bm{x}_\star^{(i,j+1)}$.

Set $[\bm{x}^{(i+1)}, \bm{\varphi}^{(i+1)}]\leftarrow$ $[\bm{x}_\star^{(i,j+1)}, \bm{\varphi}^{(i,j+1)}]$;

Update $[\bm{\alpha}^{(i+1)}, \bm{\vartheta}^{(i+1)}]$ with $[\bm{x}^{(i+1)}, \bm{\varphi}^{(i+1)}, \bm{\phi}^{(i+1)}, \bm{\rho}^{(i+1)}, \bm{\zeta}^{(i+1)}, \bm{\psi}^{(i+1)}]$;
}
Return $[\bm{x}^{(i+1)}, \bm{\varphi}^{(i+1)}, \bm{\phi}^{(i+1)}, \bm{\rho}^{(i+1)}, \bm{\zeta}^{(i+1)}, \bm{\psi}^{(i+1)}]$ as a solution $[\bm{x}^\star, \bm{\varphi}^\star, \bm{\phi}^\star, \bm{\rho}^\star, \bm{\zeta}^\star, \bm{\psi}^\star]$ to Problem $\mathbb{P}_{3}$.
\end{algorithm}
\subsection{The whole algorithm procedure}
Let the objective function value of Problem $\mathbb{P}_{i}$ be $V_{\mathbb{P}_{i}}$. We present the whole algorithm procedure in Algorithm \ref{algorithm} and the relationships between Theorems and Lemmas in \mbox{Fig. \ref{fig:relat_theorem_lemma}}. Here we summarize the overall flow of the optimization algorithm. At the $i$-th iteration, we first initialize $\bm{\alpha}^{(i-1)}, \bm{\vartheta}^{(i-1)}$ with $\bm{x}^{(i-1)}$, $\bm{\varphi}^{(i-1)}$, $\bm{\phi}^{(i-1)}$, $\bm{\rho}^{(i-1)}$, $\bm{\zeta}^{(i-1)}$, $\bm{\psi}^{(i-1)}$. Then, we fix $\bm{\alpha}, \bm{\vartheta}$ as $\bm{\alpha}^{(i-1)}, \bm{\vartheta}^{(i-1)}$ and optimize $\bm{x}$, $\bm{\varphi}$, $\bm{\phi}$, $\bm{\rho}$, $\bm{\zeta}$, $\bm{\psi}$. For the optimization of $\bm{x}$, $\bm{\varphi}$, $\bm{\phi}$, $\bm{\rho}$, $\bm{\zeta}$, $\bm{\psi}$, we use the alternative optimization technique. 

In the first step, we fix $\bm{x}$, $\bm{\varphi}$ as $\bm{x}^{(i-1)}$, $\bm{\varphi}^{(i-1)}$ and optimize $\bm{\phi}$, $\bm{\rho}$, $\bm{\zeta}$, $\bm{\psi}$. At this optimization step, we also introduce an auxiliary variable $\upsilon^{(s)}_{n,m}$, where $\upsilon^{(s)}_{n,m} = \frac{1}{2x_{n,m}\rho_np_n\varphi_nd_nr_{n,m}}$ to transform Problem $\mathbb{P}_{4}$ into a solvable concave problem $\mathbb{P}_{5}$. At the $j$-th inner iteration, we initialize $\bm{\upsilon}^{\bm{(s)}(i-1,j-1)}$ with $\bm{x}^{(i-1)}$, $\bm{\varphi}^{(i-1)}$, $\bm{\phi}^{(i-1,j-1)}$, $\bm{\rho}^{(i-1,j-1)}$, $\bm{\zeta}^{(i-1,j-1)}$, $\bm{\psi}^{(i-1,j-1)}$. We fix $\bm{\upsilon}^{\bm{(s)}}$ as $\bm{\upsilon}^{\bm{(s)}(i-1,j-1)}$ and optimize $\bm{\phi}$, $\bm{\rho}$, $\bm{\zeta}$, $\bm{\psi}$. Then we obtain the optimization results $\bm{\phi}^{(i-1,j)}$, $\bm{\rho}^{(i-1,j)}$, $\bm{\zeta}^{(i-1,j)}$, $\bm{\psi}^{(i-1,j)}$ and update $\bm{\upsilon}^{\bm{(s)}(i-1,j)}$ with these results. This optimization cycle is repeated until the difference in the objective function value of Problem $\mathbb{P}_{5}$ between the $j$-th and $(j-1)$-th iterations falls below a predefined threshold. We set the results of this alternative optimization step as $\bm{\phi}^{(i)}$, $\bm{\rho}^{(i)}$, $\bm{\zeta}^{(i)}$, $\bm{\psi}^{(i)}$.

In the second step, we fix the $\bm{\phi}$, $\bm{\rho}$, $\bm{\zeta}$, $\bm{\psi}$ as $\bm{\phi}^{(i)}$, $\bm{\rho}^{(i)}$, $\bm{\zeta}^{(i)}$, $\bm{\psi}^{(i)}$ and optimize $\bm{x}$, $\bm{\varphi}$. Then we first obtain $\bm{\varphi}^{(i)}$ and the continuous solution of $\bm{x}$ by solving Problem $\mathbb{P}_{10}$. Next, we use the Hungarian algorithm to obtain the discrete solution of $\bm{x}$ and denote it as $\bm{x}^{(i)}$.  Until now, we have obtained $\bm{x}^{(i)}$, $\bm{\varphi}^{(i)}$, $\bm{\phi}^{(i)}$, $\bm{\rho}^{(i)}$, $\bm{\zeta}^{(i)}$, $\bm{\psi}^{(i)}$. Update $\bm{\alpha}^{(i)}, \bm{\vartheta}^{(i)}$ with those results.

Repeat these two optimization steps until the difference in the objective function value of Problem $\mathbb{P}_{3}$ between the $i$-th and \mbox{$(i-1)$-th}  iterations falls in a predefined threshold. Then, we set the optimization results as $\bm{x}^\star$, $\bm{\varphi}^\star$, $\bm{\phi}^\star$, $\bm{\rho}^\star$, $\bm{\zeta}^\star$, $\bm{\psi}^\star$.

\begin{figure*}[!htbp]
\vspace{-0.3cm}
\subfigure[Convergence of the FP method.]{\includegraphics[width=.33\textwidth]{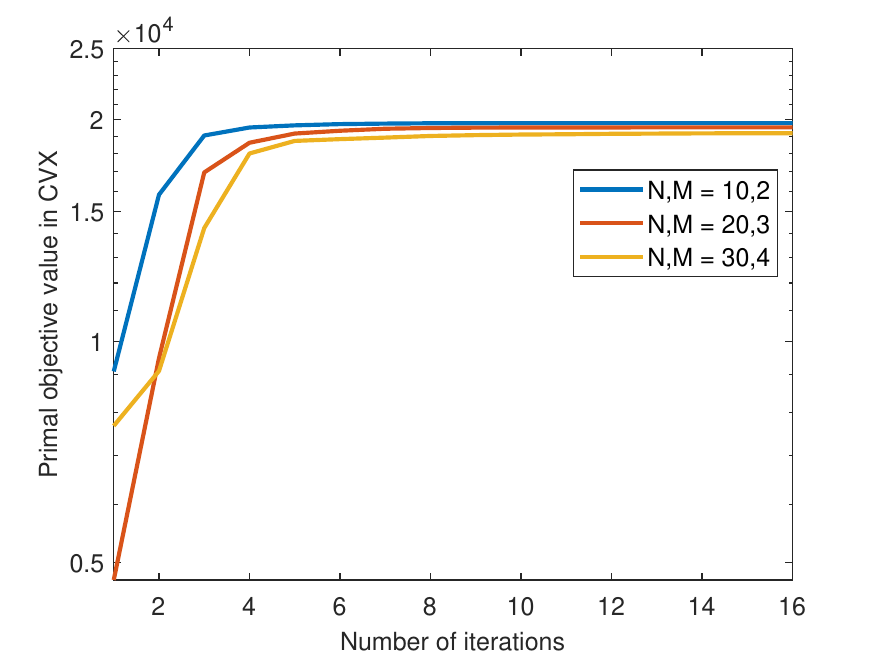}\label{fig.convergence of FP}}
\subfigure[Convergence of the QCQP method.]{\includegraphics[width=.33\textwidth]{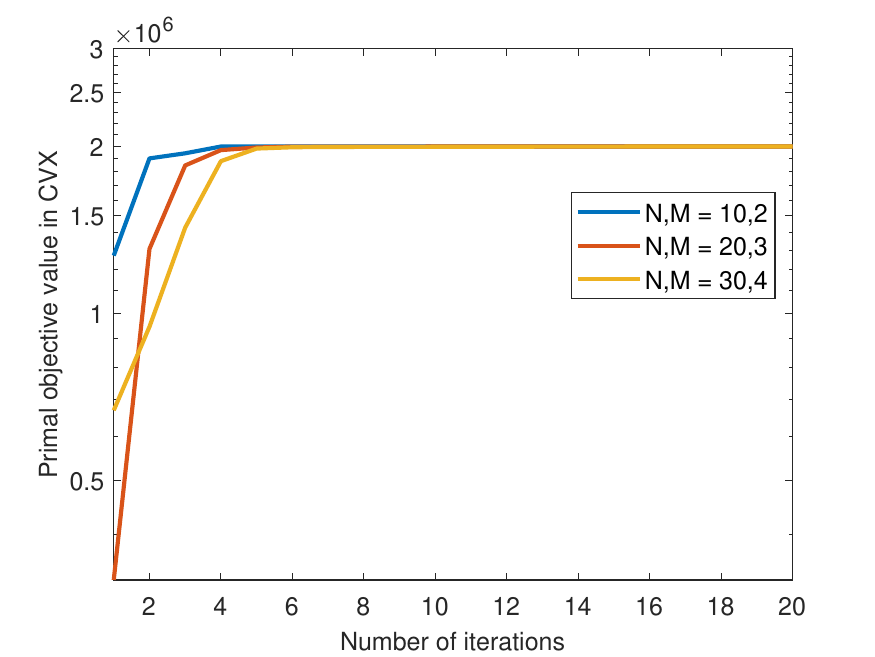}\label{fig.convergence of QCQP}}
\subfigure[Performance comparison between the DAUR algorithm and other baselines.]{\includegraphics[width=.33\textwidth]{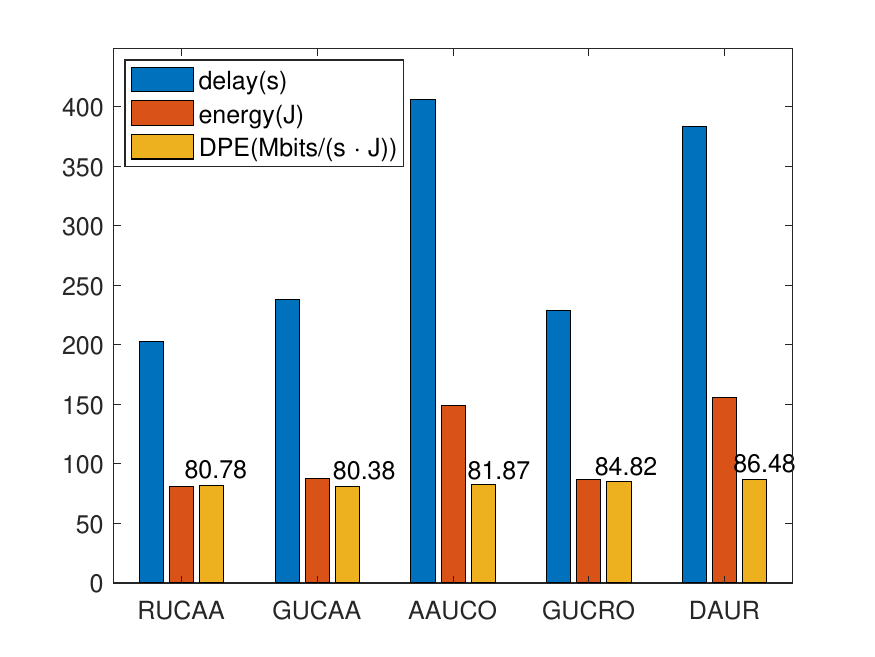}\label{fig.DAUR performance}} \vspace{-10pt}
\caption{Convergence of FP and QCQP methods; Performance comparison of the DAUR algorithm with baselines.} 
\end{figure*}
\begin{table}[!htbp]
\centering
\caption{Convergence of the DAUR Algorithm}
\label{tab:DAURConvergence}
\begin{tabular}{@{}lccc@{}}
\toprule
$N$, $M$ & \textbf{10, 2} & \textbf{20, 3} & \textbf{30, 4} \\
\midrule
Outer iter. rounds & 1 & 1 & 1 \\
QCQP rounds & 1 & 1 & 1 \\
FP rounds & 3 (26+25+25 iter.) & 2 (29+28 iter.) & 1 (30 iter.) \\
QCQP iter. & 30 & 40 & 58 \\
QCQP time & 0.17s & 11.38s & 276.42s \\
FP time & 0.03s & 0.10s & 0.08s \\
Total time & 0.20s & 11.48s & 276.50s \\
\bottomrule
\end{tabular}
\end{table}
\subsection{Novelty and wide applications of our proposed algorithm}
In this paper, we address maximizing the combined DPE of users and servers in a blockchain-enhanced Metaverse wireless system using the DAUR algorithm. This algorithm optimizes user-server associations, work offloading ratios, and task-specific computing resource distribution ratios together, as well as jointly optimizes communication and computational resources like bandwidth, transmit power, and computing allocations for both users and servers. \mbox{Unlike} previous methods that treat communication and computational resources separately~\cite{deng2022blockchain,dai2018joint}, our approach integrates them into a unified optimization problem, leading to better solutions than traditional alternating optimization methods. Additionally, the DAUR algorithm's application extends beyond DPE maximization; it's also suitable for solving energy efficiency and various utility-cost problems. For non-concave utility functions, we use successive convex approximation (SCA) to enable DAUR's application in mobile edge computing for user connection and resource allocation.
\section{Complexity Analysis}\label{sec.complexity_analysis}
In this section, we analyze the complexity of the proposed DAUR algorithm. For iterations from Line 8 to Line 12 in Algorithm \ref{algorithm}, there are $3N+3NM$ variables and $3N+3NM+2M$ constraints. The worst-case complexity of it is $\mathcal{O}((N^{3.5}+M^{3.5}+N^{3.5}M^{3.5})\log(\frac{1}{\epsilon_1}))$ with a given solution accuracy $\epsilon_1 > 0$~\cite{dai2018joint}. For iterations from Line 18 to Line 22 in Algorithm \ref{algorithm}, there are $N+NM$ variables and $NM+2N+2M+2$ constraints. The worst-case complexity of it is $\mathcal{O}((N^{3.5}+M^{3.5}+N^{3.5}M^{3.5})\log(\frac{1}{\epsilon_2}))$ with a given solution accuracy $\epsilon_2 > 0$. The complexity of the Hungarian algorithm is $\mathcal{O}(N^3M^3)$~\cite{kuhn1955hungarian}. To summarize, if Algorithm \ref{algorithm} takes $\mathcal{I}$ iterations, the whole complexity is \mbox{$\mathcal{I}\times\mathcal{O}((N^{3.5}+M^{3.5}+N^{3.5}M^{3.5})\log(\frac{1}{\epsilon_3}))$} with a given solution accuracy $\epsilon_3 > 0$.

\begin{figure*}[t]
\vspace{-0.3cm}
\subfigure[Performance comparisons of different total bandwidth.]{\includegraphics[width=.33\textwidth]{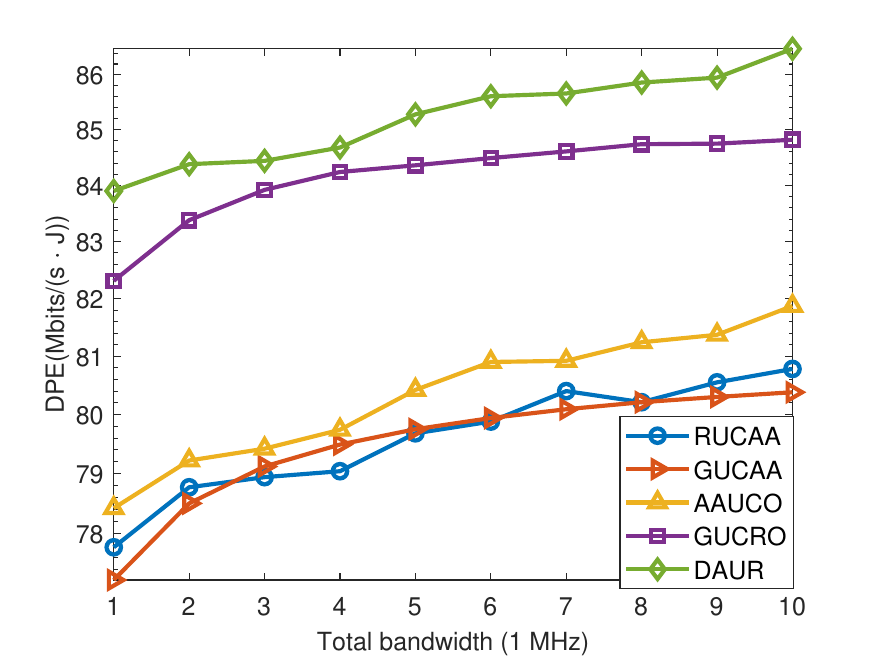}\label{fig.comparison bandwidth}}
\subfigure[Performance comparisons of different server frequency.]{\includegraphics[width=.33\textwidth]{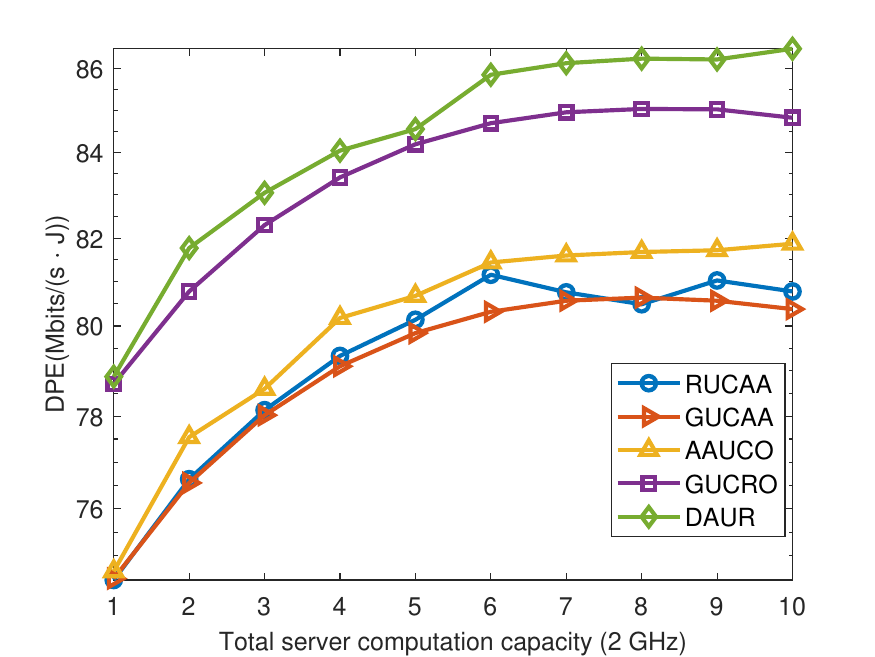}\label{fig.comparison server frequency}}
\subfigure[Performance comparisons of different user frequency.]{\includegraphics[width=.33\textwidth]{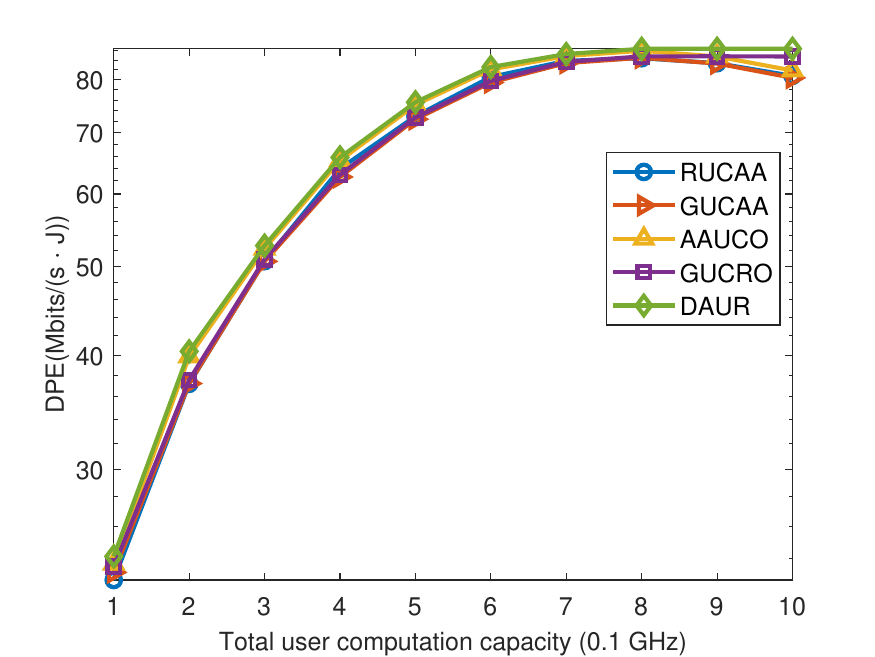}\label{fig.comparison user frequency}}
\subfigure[Performance comparisons of different user transmit power.]{\includegraphics[width=.33\textwidth]{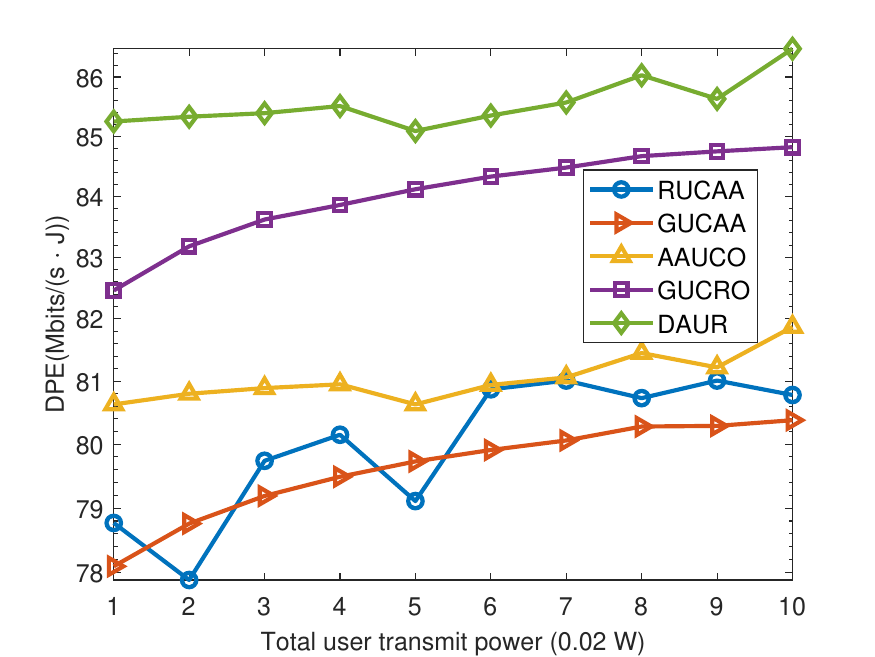}\label{fig.comparison user transmit power}} 
\subfigure[Performance comparisons of different weight ratios.]{\includegraphics[width=.33\textwidth]{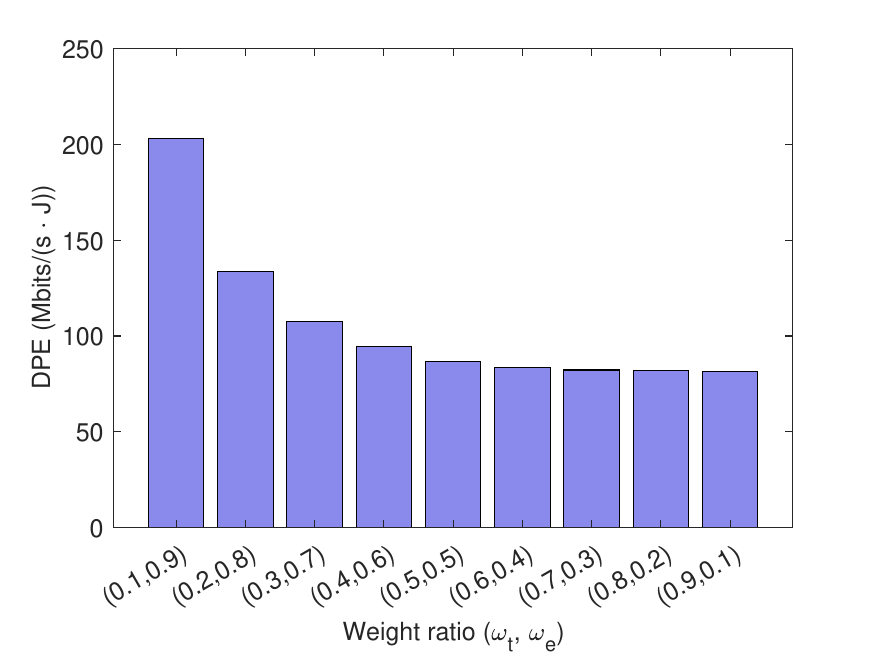}\label{fig.comparison omega}} 
\subfigure[Performance comparisons of different preference weights.]{\includegraphics[width=.33\textwidth]{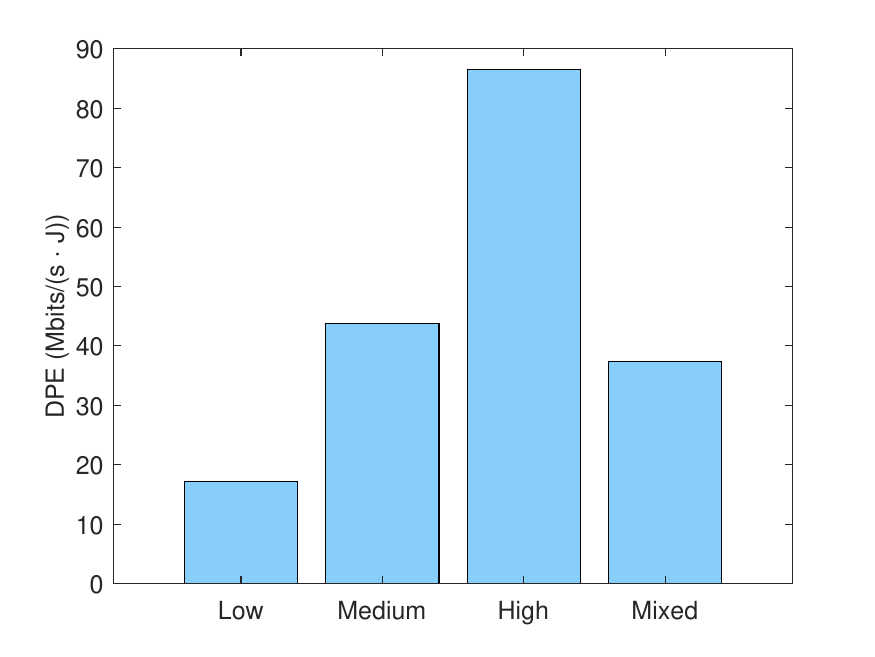}\label{fig.comparison preference}} 
\vspace{-10pt}
\caption{Performance comparison of different communication and computation resources and heterogeneous settings.}
\end{figure*}
\section{Simulation Results}\label{sec.simulation_results}
In this section, we give the default simulation settings and then present the numerical results and analysis.

\textbf{Default settings.}
We employ a circular network topology with a radius of 1000 meters, comprising ten users and two servers randomly positioned. The large-scale fading between user $n$ and server $m$, denoted as $h_{n,m}$, follows the model $128.1+37.6 \log_{10}d_{n,m}$, with $d_{n,m}$ representing the Euclidean distance between them. Rayleigh fading is assumed for small-scale fading. The Gaussian noise power spectral density $\sigma^2$, is set at $-134$dBm, and each server's total bandwidth $b_m$, is 10 MHz. Mobile users have a maximum transmit power $p_{n}$, of 0.2 W, and their maximum computational capacity $f_{n}$, is 1 GHz, while servers' $f_{m}$, have a capacity of 20 GHz. Both mobile users and servers process data at 279.62 CPU cycles per bit ($\eta_n$ and $\eta_m$), and servers require 737.5 CPU cycles per bit for blockchain block generation. The effective switched capacitance for both $\kappa_n$ and $\kappa_m$ is $10^{-27}$. Data sizes for mobile users vary between 500 KB and 2000 KB, determined by scaling uniformly distributed pseudorandom values within this range. The block size $S_b$ is 8 MB, and the data rate for wired server links $R_{m}$ is 15 Mbps. The ratio of data size change between MEC and blockchain tasks $\omega_b$ is fixed at 1. For delay and energy consumption parameters ($\omega_t$ and $\omega_e$), values of 0.5 each are used. DPE preference weights, $c_n$ and $c_{n,m}$ are set at $1 \times \frac{1}{5\times10^{5}}$ to maintain DPE values between 1 and 100. The Mosek CVX optimization tool is utilized for the simulations.

\textbf{Convergence of the proposed DAUR algorithm.}
In Fig. \ref{fig.convergence of FP} and Fig. \ref{fig.convergence of QCQP}, when examining the FP and QCQP methods individually, they demonstrate convergence within just 8 iterations. In Table \ref{tab:DAURConvergence}, the DAUR algorithm showcases effective convergence across various user-server configurations in a blockchain-enabled Metaverse wireless communication system. For the (10, 2) user-server setup, it achieves convergence in just one total outer iteration, utilizing one QCQP round and three FP rounds. The QCQP method takes 30 iterations and 0.17 seconds, while the FP method cumulatively takes 76 iterations (computed from 26+25+25) and 0.03 seconds, leading to a total time of 0.20 seconds. In a more complex (20, 3) configuration, the algorithm still maintains its efficiency with one total outer iteration, one QCQP round of 40 iterations (11.38 seconds), and two FP rounds totaling 57 iterations (0.10 seconds), culminating in 11.48 seconds overall. For the largest tested setup of (30, 4), the DAUR algorithm consistently exhibits its robust convergence capability, completing within one outer iteration, one QCQP round of 58 iterations (276.42 seconds), and one FP round of 30 iterations (0.08 seconds), summing up to 276.50 seconds in total.

\textbf{Performance comparison with other baselines.}
We choose the following baselines --- RUCAA: random user connection with average resource
allocation; GUCAA: greedy user connection with average resource
allocation; AAUCO: average resource allocation with user connection
optimization; GUCRO: greedy user connection with resource allocation optimization. Note that user connection optimization and resource allocation refer to the QCQP and FP methods, respectively.

The DAUR algorithm outperforms other baseline methods in DPE performance in Fig. \ref{fig.DAUR performance}, achieving 86.48 M bits/(s $\cdot$ J). This surpasses RUCAA (80.78), GUCAA (80.38), AAUCO (81.87), and GUCRO (84.82), demonstrating its superior efficiency in resource allocation and user connection optimization, attributed to its effective integration of QCQP and FP methods.

\textbf{Performance comparison of different communication and computation resources.}
Across the server available bandwidth range of 1 to 10 MHz, the DAUR algorithm consistently outshines other methods in DPE performance in Fig. \ref{fig.comparison bandwidth}, starting at 83.90 and leading a peak of 86.48, notably surpassing other methods in efficiency enhancement. With the increase in communication bandwidth, the DPE of each method shows an increasing trend, which indicates that under the existing configuration, DPE is positively correlated with communication bandwidth.

In the simulation with varying server computation capacities from 2 GHz to 20 GHz in Fig. \ref{fig.comparison server frequency}, the DAUR algorithm exhibits superior DPE performance, steadily increasing from 78.87 to 86.48. This outperformance is consistent across all capacity levels, surpassing the other methods.
In the simulation where user computing capacity varies from 0.1 GHz to 1 GHz in Fig. \ref{fig.comparison user frequency}, the DAUR algorithm consistently demonstrates the highest DPE performance, increasing from 24.10 to 86.48. This trend indicates DAUR's effective handling of increased computing capacities, outperforming other methods at every step.
In summary, the DAUR algorithm outperforms other baselines in varying computing capacities, demonstrating its effective DPE performance and resource management in both server and user capacity simulations.

In Fig. \ref{fig.comparison user transmit power} where the maximum transmit power of mobile users ($p_n$) is adjusted from 0.02 W to 0.2 W, the DAUR algorithm demonstrates superior and consistent DPE performance, starting at 85.25 and peaking at 86.48. 
From Figs. \ref{fig.comparison server frequency}, \ref{fig.comparison user frequency}, and \ref{fig.comparison user transmit power}, we know that with the increase of computation capacity or user transmit power, there may be a decrease in DPE because both DAUR and baseline methods can only obtain sub-optimal solutions. In this case, we can get a better solution by changing the initial point. Second, DAUR's approach is still better than the baselines using partial optimization, indicating that changes in network communication and computing resources require us to re-optimize user connectivity and resource allocation simultaneously.

\textbf{Performance comparison of heterogeneous settings.}
When adjusting the balance between delay and energy consumption weights ($\omega_t$ and $\omega_e$) in the DAUR algorithm, a notable pattern emerges in DPE values in Fig. \ref{fig.comparison omega}. As the emphasis shifts progressively from delay-centric $(0.1, 0.9)$ to energy-centric $(0.9, 0.1)$, there's an initial steep decrease in DPE from 203.15 to a balanced point of 86.48 at $(0.5, 0.5)$, followed by a more gradual reduction, eventually reaching 81.3566. This trend underscores the algorithm's sensitivity to the trade-off between delay and energy efficiency, significantly influencing its performance.

We consider four preference parameter setting cases: 1. low preference: set $c_n$ and $c_{n,m}$ as $0.2/(5\times10^{5})$;
2. medium preference: set $c_n$ and $c_{n,m}$ as $0.5/(5\times10^{5})$;
3. high preference: set $c_n$ and $c_{n,m}$ as $1/(5\times10^{5})$;
4. mixed preference: set $c_n$ and $c_{n,m}$ as $a/(5\times10^{5})$,
where $a$ is a random value uniformly taken from [0, 1]. Fig. \ref{fig.comparison preference} presents the results of the DAUR algorithm under different preference settings for a network with 10 users and 2 servers. In the low preference scenario, the DPE is at its lowest (17.18), indicating minimal efficiency. The medium preference setting improves to a DPE of 43.83, while the high preference achieves the best performance at 86.48, showing the highest efficiency. The mixed preference setting yields a moderate DPE of 37.42, indicating balanced resource distribution. These results show how different preference intensities directly impact the system's data processing efficiency.

\section{Conclusion and Future Direction}\label{sec.conclusion}
In conclusion, this paper presents a potential metric in blockchain-empowered Metaverse wireless communication systems, the concept of data processing efficiency (DPE), and develops a novel DAUR algorithm for optimizing user association and resource allocation. Our approach can transform complex DPE optimization into convex problems, demonstrating superior efficiency over traditional methods through extensive numerical analysis. The potential application of this work is that the Metaverse server helps users conduct non-fungible token (NFT) tasks and considers the maximization of data processing efficiency. In the current work, the algorithm is still a centralized method, which may increase the risk of exposing the users' and servers' private information. We will consider a decentralized way to operate the DAUR algorithm in the future.

\section*{Acknowledgement}

This research is supported by the National Research Foundation, Singapore under its Strategic Capability Research Centres Funding Initiative, Singapore MOE AcRF Tier 1 RT5/23, Tier 1 RG90/22, and NTU-WASP Joint Project. Any opinions, findings and conclusions or recommendations expressed in this material are those of the author(s) and do not reflect the views of National Research Foundation, Singapore.

\bibliographystyle{ACM-Reference-Format}
\bibliography{ref}

\clearpage
\begin{appendices}
\section{Proof of \textbf{Lemma \ref{lemma_p1top2}}}\label{append_lemma_p1top2}
\begin{proof}
Define new auxiliary variables $\vartheta_n^{(u)}$ and $\vartheta_{n,m}^{(s)}$. Let $\vartheta_n^{(u)} \leq \frac{c_n(1-\varphi_n)d_n}{cost^{(u)}_n}$ and $\vartheta_{n,m}^{(s)} \leq \frac{c_{n,m}x_{n,m}\varphi_nd_n}{cost^{(s)}_{n,m}}$. According to Equations (\text{\ref{eq.cost_u}}) and (\text{\ref{eq.cost_s}}), we obtain that
\begin{talign}
    &cost^{(u)}_n = \omega_t T^{(up)}_n + \omega_e E^{(up)}_n, cost^{(u)}_n \leq \frac{c_n(1-\varphi_n)d_n}{\vartheta_n^{(u)}} \nonumber \\
    &\Rightarrow \omega_t T^{(up)}_n + \omega_e E^{(up)}_n \leq \frac{c_n(1-\varphi_n)d_n}{\vartheta_n^{(u)}} \nonumber \\
    &\Rightarrow \omega_t T^{(up)}_n + \omega_e \kappa_n(1-\varphi_n)d_n\eta_n(\psi_nf_n)^2 - \frac{c_n(1-\varphi_n)d_n}{\vartheta_n^{(u)}}\leq 0, \label{eq.cost_u2}
\end{talign}
\begin{talign}
    &cost^{(s)}_{n,m} = \omega_t (T^{(ut)}_{n,m} + T^{(sp)}_{n,m} + T^{(sg)}_{n,m} + T^{(bp)}_{n,m} + T^{(sv)}_{n,m}) \nonumber \\ 
    &+ \omega_e (E^{(ut)}_{n,m} + E^{(sp)}_{n,m} + E^{(sg)}_{n,m}),  cost^{(s)}_{n,m} \leq \frac{c_{n,m}x_{n,m}\varphi_nd_n}{\vartheta_{n,m}^{(s)}} \nonumber \\
    &\Rightarrow \omega_t (T^{(ut)}_{n,m} + T^{(sp)}_{n,m} + T^{(sg)}_{n,m} + T^{(bp)}_{n,m} + T^{(sv)}_{n,m}) \nonumber \\ 
    &+ \omega_e (E^{(ut)}_{n,m} + E^{(sp)}_{n,m} + E^{(sg)}_{n,m})\leq \frac{c_{n,m}x_{n,m}\varphi_nd_n}{\vartheta_{n,m}^{(s)}} \nonumber \\
    &\Rightarrow \omega_t(T^{(ut)}_{n,m} + T^{(sp)}_{n,m} + T^{(sg)}_{n,m} + T^{(bp)}_{n,m} + T^{(sv)}_{n,m}) \nonumber \\ 
    &+ \omega_e \{\frac{x_{n,m}\rho_np_n\varphi_nd_n}{r_{n,m}} + \kappa_mx_{n,m}\varphi_nd_n\eta_m(\gamma_{n,m}\zeta_{n,m}f_m)^2 \nonumber \\ 
    &+\!\! \kappa_mx_{n,m}\varphi_nd_n\eta_m \omega_b [(1\!\!-\!\!\gamma_{n,m})\zeta_{n,m}f_m]^2\} \!\!-\!\! \frac{c_{n,m}x_{n,m}\varphi_nd_n}{\vartheta_{n,m}^{(s)}} \leq 0,\label{eq.cost_s2}
\end{talign}
where ``$\Rightarrow$'' means ``imply''. Next, we leverage two new auxiliary variables $T^{(u)}_n$ and $T^{(s)}_{n,m}$ to replace the delay in Equations (\text{\ref{eq.cost_u2}}) and (\text{\ref{eq.cost_s2}}). With the introduction of these two variables, there would be four more new constraints:
\begin{talign}
    &\omega_t T^{(u)}_n + \omega_e \kappa_n(1-\varphi_n)d_n\eta_n(\psi_nf_n)^2 - \frac{c_n(1-\varphi_n)d_n}{\vartheta_n^{(u)}}\leq 0,\\
    &\omega_t T^{(s)}_{n,m} \!+ \!\omega_e \{\frac{x_{n,m}\rho_np_n\varphi_nd_n}{r_{n,m}} \!+ \!\kappa_mx_{n,m}\varphi_nd_n\eta_m(\gamma_{n,m}\zeta_{n,m}f_m)^2 \nonumber \\ 
    &+ \kappa_mx_{n,m}\varphi_nd_n\eta_m \omega_b [(1-\gamma_{n,m})\zeta_{n,m}f_m]^2\} \nonumber \\
    &- \frac{c_{n,m}x_{n,m}\varphi_nd_n}{\vartheta_{n,m}^{(s)}} \leq 0, \\
    &T^{(up)}_n \leq T^{(u)}_n, \\
    &T^{(ut)}_{n,m} + T^{(sp)}_{n,m} + T^{(sg)}_{n,m} + T^{(bp)}_{n,m} + T^{(sv)}_{n,m} \leq T^{(s)}_{n,m}.
\end{talign}
Let $\bm{T}:=\{\bm{T^{(u)}},\bm{T^{(s)}}\}$ and $\bm{\vartheta}:=\{\bm{\vartheta^{(u)}},\bm{\vartheta^{(s)}}\}$. Then, the Problem $\mathbb{P}_{1}$ can be transformed into the \mbox{Problem $\mathbb{P}_{2}$}.

Thus, \textbf{Lemma \ref{lemma_p1top2}} is proven.
\end{proof}

\section{Proof of \textbf{Lemma \ref{lemma_p2top3}}}\label{append_lemma_p2top3}
\begin{proof}
Here we analyze part of the KKT condition of Problem $\mathbb{P}_{2}$ to facilitate our subsequent analysis. We introduce non-negative variables $\alpha^{(u)}_n$ and $\alpha^{(s)}_{n,m}$ as the multipliers. Let $\bm{\alpha^{(u)}}:=[\alpha^{(u)}_n]|_{n\in\mathcal{N}}$, $\bm{\alpha^{(s)}}:=[\alpha^{(s)}_{n,m}]|_{n\in\mathcal{N},m\in\mathcal{M}}$, and $\bm{\alpha}:=\{\bm{\alpha^{(u)}}, \bm{\alpha^{(s)}}\}$. The Lagrangian function is given as follows:
\begin{talign}
    &L_{\mathbb{P}_{2}}(\bm{x},\bm{\varphi},\bm{\gamma},\bm{\phi},\bm{\rho},\bm{\zeta},\bm{\psi}, \bm{\vartheta}, \bm{T}, \bm{\alpha}) \nonumber \\
    &= - \sum_{n \in \mathcal{N}} \vartheta_n^{(u)} -\sum_{n \in \mathcal{N}} \sum_{m \in \mathcal{M}} \vartheta_{n,m}^{(s)} \nonumber \\
    &+ \sum_{n \in \mathcal{N}} \alpha^{(u)}_n \cdot [\vartheta_n^{(u)}cost^{(u)}_n - c_n(1 - \varphi_n)d_n] \nonumber \\
    &+ \sum_{n \in \mathcal{N}}\sum_{m \in \mathcal{M}}\alpha^{(s)}_{n,m} \cdot [\vartheta_{n,m}^{(s)}cost^{(s)}_{n,m} - c_{n,m}x_{n,m}\varphi_n d_n] \nonumber \\
    &+ \hat{L}_{\mathbb{P}_{2}},
\end{talign}
where $\hat{L}_{\mathbb{P}_{2}}$ is the remaining Lagrangian terms that we don't care about. Next, we analyze some stationarity and complementary slackness properties of $L_{\mathbb{P}_{2}}$.

\textbf{Stationarity:}
\begin{talign}
    &\frac{\partial L_{\mathbb{P}_{2}}}{\partial \vartheta_n^{(u)}} = -1 + \alpha^{(u)}_n cost^{(u)}_n=0, \forall n \in \mathcal{N},\\
    &\frac{\partial L_{\mathbb{P}_{2}}}{\partial \vartheta_{n,m}^{(s)}} = -1 + \alpha^{(s)}_{n,m} cost^{(s)}_{n,m}=0, \forall n \in \mathcal{N}, \forall m \in \mathcal{M}.
\end{talign}

\textbf{Complementary slackness:}
\begin{talign}
    &\alpha^{(u)}_n \cdot [\vartheta_n^{(u)}cost^{(u)}_n - c_n(1 - \varphi_n)d_n] = 0,\forall n \in \mathcal{N},\\
    &\alpha^{(s)}_{n,m} \cdot (\vartheta_{n,m}^{(s)}cost^{(s)}_{n,m} - c_{n,m}x_{n,m}\varphi_n d_n) = 0, \nonumber \\
    &\forall n \in \mathcal{N}, \forall m \in \mathcal{M}.
\end{talign}
Thus, at KKT points of Problem $\mathbb{P}_{2}$, we can obtain that 
\begin{talign}
    &\alpha^{(u)}_n = \frac{1}{cost^{(u)}_n},\\
    &\alpha^{(s)}_{n,m} = \frac{1}{cost^{(s)}_{n,m}},\\
    &\vartheta_n^{(u)} = \frac{c_n(1 - \varphi_n)d_n}{cost^{(u)}_n},\\
    &\vartheta_{n,m}^{(s)} = \frac{c_{n,m}x_{n,m}\varphi_n d_n}{cost^{(s)}_{n,m}}.
\end{talign}
Based on the above discussion, Problem $\mathbb{P}_{2}$ can be transformed into Problem $\mathbb{P}_{3}$.

\textbf{Lemma \ref{lemma_p2top3}} holds.
\end{proof}

\section{Proof of \textbf{Lemma \ref{lemma_p4top5}}}\label{append_lemma_p4top5}
\begin{proof}
We introduce one auxiliary variable $\upsilon^{(s)}_{n,m}$, where
\begin{talign}
    \upsilon^{(s)}_{n,m} = \frac{1}{2x_{n,m}\rho_np_n\varphi_nd_nr_{n,m}}.
\end{talign}
Then, $cost^{(s)}_{n,m}$ can be rewritten as:
\begin{talign}
    &\widetilde{cost}^{(s)}_{n,m} = \omega_t T^{(s)}_{n,m} + \omega_e \{(x_{n,m}\rho_np_n\varphi_nd_n)^2\upsilon^{(s)}_{n,m} + \frac{1}{4r_{n,m}^2\upsilon^{(s)}_{n,m}}\nonumber \\ 
    &+ \kappa_mx_{n,m}\varphi_nd_n\eta_m(\gamma_{n,m}\zeta_{n,m}f_m)^2 \nonumber \\ 
    &+ \kappa_mx_{n,m}\varphi_nd_n\eta_m \omega_b [(1-\gamma_{n,m})\zeta_{n,m}f_m]^2\}. \label{eq.vartheta_s_new}
\end{talign}
Since $x_{n,m}, \varphi_n, \gamma_{n,m}, \upsilon^{(s)}_{n,m}$ is given, Equation (\text{\ref{eq.vartheta_s_new}}) is a convex function.

Let $\chi(\rho_n) = x_{n,m}\rho_np_n\varphi_nd_n$ and $\varsigma(\phi_{n,m},\rho_n) = r_{n,m}$, where $r_{n,m} = \phi_{n,m}b_m\log_2(1 + \frac{\rho_np_ng_{n,m}}{\sigma^2\phi_{n,m}b_m})$. It's easy to know that $\chi(\rho_n)$ is convex of $\rho_n$ and $\varsigma(\phi_{n,m},\rho_n)$ is jointly concave of $(\phi_{n,m},\rho_n)$. We define two following functions:
\begin{talign}
    &\mathcal{F}(\rho_n,\phi_{n,m},\zeta_{n,m},T^{(s)}_{n,m}) = \omega_t T^{(s)}_{n,m} + \omega_e \{\frac{\chi(\rho_n)}{\varsigma(\phi_{n,m},\rho_n)} \nonumber \\ 
    &+ \kappa_mx_{n,m}\varphi_nd_n\eta_m(\gamma_{n,m}\zeta_{n,m}f_m)^2 \nonumber \\ 
    &+ \kappa_mx_{n,m}\varphi_nd_n\eta_m \omega_b [(1-\gamma_{n,m})\zeta_{n,m}f_m]^2\},
\end{talign}
\begin{talign}
    &\mathcal{G}(\rho_n,\phi_{n,m},\zeta_{n,m},T^{(s)}_{n,m}) = \omega_t T^{(s)}_{n,m} \nonumber \\ 
    &+ \omega_e \{\chi(\rho_n)^2\upsilon^{(s)}_{n,m} + \frac{1}{4\varsigma(\phi_{n,m},\rho_n)^2\upsilon^{(s)}_{n,m}}\nonumber \\ 
    &+ \kappa_mx_{n,m}\varphi_nd_n\eta_m(\gamma_{n,m}\zeta_{n,m}f_m)^2 \nonumber \\ 
    &+ \kappa_mx_{n,m}\varphi_nd_n\eta_m \omega_b [(1-\gamma_{n,m})\zeta_{n,m}f_m]^2\}.
\end{talign}
The partial derivative of $T^{(s)}_{n,m}$ is given by
\begin{talign}
    \frac{\partial\mathcal{F}(\rho_n,\phi_{n,m},\zeta_{n,m},\vartheta^{(s)}_{n,m},T^{(s)}_{n,m})}{\partial T^{(s)}_{n,m}} = \omega_t,\\
    \frac{\partial\mathcal{G}(\rho_n,\phi_{n,m},\zeta_{n,m},\vartheta^{(s)}_{n,m},T^{(s)}_{n,m})}{\partial T^{(s)}_{n,m}} = \omega_t.
\end{talign}
We can easily get that
\begin{talign}
    \frac{\partial\mathcal{F}(\rho_n,\phi_{n,m},\zeta_{n,m},\vartheta^{(s)}_{n,m},T^{(s)}_{n,m})}{\partial T^{(s)}_{n,m}} = \frac{\partial\mathcal{G}(\rho_n,\phi_{n,m},\zeta_{n,m},\vartheta^{(s)}_{n,m},T^{(s)}_{n,m})}{\partial T^{(s)}_{n,m}}.
\end{talign}
The partial derivative of $\zeta_{n,m}$ is
\begin{talign}
    &\frac{\partial \mathcal{F}(\rho_n,\phi_{n,m},\zeta_{n,m},\vartheta^{(s)}_{n,m},T^{(s)}_{n,m})}{\partial \zeta_{n,m}} = 2\kappa_mx_{n,m}\varphi_nd_n\eta_m\zeta_{n,m}(\gamma_{n,m}f_m)^2 \nonumber \\
    &+ 2\kappa_mx_{n,m}\varphi_nd_n\eta_m \omega_b \zeta_{n,m}[(1-\gamma_{n,m})f_m]^2,\\
    &\frac{\partial \mathcal{G}(\rho_n,\phi_{n,m},\zeta_{n,m},\vartheta^{(s)}_{n,m},T^{(s)}_{n,m})}{\partial \zeta_{n,m}} = 2\kappa_mx_{n,m}\varphi_nd_n\eta_m\zeta_{n,m}(\gamma_{n,m}f_m)^2 \nonumber \\
    &+ 2\kappa_mx_{n,m}\varphi_nd_n\eta_m \omega_b \zeta_{n,m}[(1-\gamma_{n,m})f_m]^2.
\end{talign}
We find that
\begin{talign}
    \frac{\partial \mathcal{F}(\rho_n,\phi_{n,m},\zeta_{n,m},\vartheta^{(s)}_{n,m},T^{(s)}_{n,m})}{\partial \zeta_{n,m}} = \frac{\partial \mathcal{G}(\rho_n,\phi_{n,m},\zeta_{n,m},\vartheta^{(s)}_{n,m},T^{(s)}_{n,m})}{\partial \zeta_{n,m}}.
\end{talign}
The partial derivative of $\rho_n$ is shown as follows:
\begin{talign}
    &\frac{\partial \mathcal{F}(\rho_n,\phi_{n,m},\zeta_{n,m},\vartheta^{(s)}_{n,m},T^{(s)}_{n,m})}{\partial \rho_n} \nonumber \\ 
    &= \omega_e\frac{\frac{\partial \chi(\rho_n)}{\partial \rho_n}\varsigma(\phi_{n,m},\rho_n) - \chi(\rho_n) \frac{\partial \varsigma(\phi_{n,m},\rho_n)}{\partial \rho_n}}{\varsigma(\phi_{n,m},\rho_n)^2},\\
    &\frac{\partial \mathcal{G}(\rho_n,\phi_{n,m},\zeta_{n,m},\vartheta^{(s)}_{n,m},T^{(s)}_{n,m})}{\partial \rho_n} \nonumber \\
    &= \omega_e (2\upsilon^{(s)}_{n,m} \chi(\rho_n) \frac{\partial \chi(\rho_n)}{\partial \rho_n} - \frac{1}{2\upsilon^{(s)}_{n,m} \varsigma(\phi_{n,m},\rho_n)^3}\frac{\partial \varsigma(\phi_{n,m},\rho_n)}{\partial \rho_n}).
\end{talign}
When $\upsilon^{(s)}_{n,m} = \frac{1}{2x_{n,m}\rho_np_n\varphi_nd_nr_{n,m}}$, we can get that
\begin{talign}
    \frac{\partial \mathcal{F}(\rho_n,\phi_{n,m},\zeta_{n,m},\vartheta^{(s)}_{n,m},T^{(s)}_{n,m})}{\partial \rho_n} = \frac{\partial \mathcal{G}(\rho_n,\phi_{n,m},\zeta_{n,m},\vartheta^{(s)}_{n,m},T^{(s)}_{n,m})}{\partial \rho_n}.
\end{talign}
The partial derivative of $\phi_{n,m}$ is
\begin{talign}
    &\frac{\partial \mathcal{F}(\rho_n,\phi_{n,m},\zeta_{n,m},\vartheta^{(s)}_{n,m},T^{(s)}_{n,m})}{\partial \phi_{n,m}} = -\frac{\omega_e\chi(\rho_n)}{\varsigma(\phi_{n,m},\rho_n)^2}\frac{\partial \varsigma(\phi_{n,m},\rho_n)}{\partial \phi_{n,m}},\\
    &\frac{\partial \mathcal{G}(\rho_n,\phi_{n,m},\zeta_{n,m},\vartheta^{(s)}_{n,m},T^{(s)}_{n,m})}{\partial \phi_{n,m}} = -\frac{\omega_e}{2\upsilon^{(s)}_{n,m}\varsigma(\phi_{n,m},\rho_n)^3}\frac{\partial \varsigma(\phi_{n,m},\rho_n)}{\partial \phi_{n,m}}.
\end{talign}
When $\upsilon^{(s)}_{n,m} = \frac{1}{2x_{n,m}\rho_np_n\varphi_nd_nr_{n,m}}$, we can obtain that
\begin{talign}
    \frac{\partial \mathcal{F}(\rho_n,\phi_{n,m},\zeta_{n,m},\vartheta^{(s)}_{n,m},T^{(s)}_{n,m})}{\partial \phi_{n,m}} =\frac{\partial \mathcal{G}(\rho_n,\phi_{n,m},\zeta_{n,m},\vartheta^{(s)}_{n,m},T^{(s)}_{n,m})}{\partial \phi_{n,m}}.
\end{talign}
Based on the above discussion, we can obtain that
\begin{talign}
     \frac{\partial \mathcal{F}(\rho_n,\phi_{n,m},\zeta_{n,m},\vartheta^{(s)}_{n,m},T^{(s)}_{n,m})}{\partial (\rho_n,\phi_{n,m},\zeta_{n,m},\vartheta^{(s)}_{n,m},T^{(s)}_{n,m})} =\frac{\partial \mathcal{G}(\rho_n,\phi_{n,m},\zeta_{n,m},\vartheta^{(s)}_{n,m},T^{(s)}_{n,m})}{\partial (\rho_n,\phi_{n,m},\zeta_{n,m},\vartheta^{(s)}_{n,m},T^{(s)}_{n,m})}.
\end{talign}
Besides, it's easy to know that 
\begin{talign}
    &\mathcal{F}(\rho_n,\phi_{n,m},\zeta_{n,m},\vartheta^{(s)}_{n,m},T^{(s)}_{n,m})\nonumber \\ 
    &= \mathcal{G}(\rho_n,\phi_{n,m},\zeta_{n,m},\vartheta^{(s)}_{n,m},T^{(s)}_{n,m}),
\end{talign}
when $\upsilon^{(s)}_{n,m} = \frac{1}{2x_{n,m}\rho_np_n\varphi_nd_nr_{n,m}}$. Therefore, the function of $\mathcal{F}(\cdot)$ is the same as that of $\mathcal{G}(\cdot)$. Let $\bm{\upsilon^{(s)}} := [\upsilon^{(s)}_{n,m}|_{\forall n \in \mathcal{N},\forall m \in \mathcal{M}}]$. Problem $\mathbb{P}_{4}$ can be transformed into Problem $\mathbb{P}_{5}$. In Problem $\mathbb{P}_{5}$, if given $\bm{\upsilon^{(s)}}$, it would be a convex optimization problem.

At the $i$-th iteration, if we first fix $\bm{\upsilon}^{\bm{(s)}(i-1)}$, Problem $\mathbb{P}_{5}$ would be a concave optimization problem. Then we optimize $\bm{\phi}^{(i)},\bm{\rho}^{(i)},\bm{\zeta}^{(i)}$, $\bm{\psi}^{(i)}, \bm{T}^{(i)}$. After we obtain the results of them, we then update $\bm{\upsilon}^{\bm{(s)}(i)}$ according to those results. Because the alternative optimization of the Problem $\mathbb{P}_{5}$ is non-decreasing, as $i\rightarrow\infty$, we can finally obtain the optimal solutions of Problem $\mathbb{P}_{5}$ (i.e., $\bm{\phi}^{(\star)}$, $\bm{\rho}^{(\star)}$, $\bm{\zeta}^{(\star)}$, $\bm{\psi}^{(\star)}$, $\bm{T}^{(\star)}$, $\bm{\upsilon}^{\bm{(s)}(\star)}$). We know that $\upsilon^{(s)}_{n,m} = \frac{1}{2x_{n,m}\rho_np_n\varphi_nd_nr_{n,m}}$. Thus, with $\bm{\upsilon}^{\bm{(s)}(\star)}$, we can find $\bm{\phi}^{(\star)},\bm{\rho}^{(\star)},\bm{\zeta}^{(\star)},\bm{\psi}^{(\star)}, \bm{T}^{(\star)}$, which is a stationary point of Problem $\mathbb{P}_{5}$.

\textbf{Lemma \ref{lemma_p4top5}} is proven.
\end{proof}

\section{Proof of \textbf{Lemma \ref{lemma_gamma}}}\label{append_lemma_gamma}
\begin{proof}
    We first analyze $\vartheta_{n,m}^{(s)}cost_{n,m}^{(s)}$, and write the explicit expression of it:
\begin{talign}
    &\vartheta_{n,m}^{(s)}cost_{n,m}^{(s)} \nonumber \\
    &= \vartheta_{n,m}^{(s)}\omega_t T^{(s)}_{n,m} + \vartheta_{n,m}^{(s)}\omega_e \frac{\rho_n p_n d_n}{r_{n,m}}x_{n,m}\varphi_n \nonumber \\
    &+ \vartheta_{n,m}^{(s)}\omega_e \{\kappa_m d_n \eta_m\zeta_{n,m}^2 f_m^2 [\gamma_{n,m}^2 + \omega_b (1-\gamma_{n,m})^2]x_{n,m}\varphi_n\},
\end{talign}
where $\gamma_{n,m}^2 + \omega_b (1-\gamma_{n,m})^2$ is independent of the others except $T^{(s)}_{n,m}$. when $\gamma_{n,m} = \frac{\omega_b}{1+\omega_b}$, $\gamma_{n,m}^2 + \omega_b (1-\gamma_{n,m})^2$ takes the minimum value. In $T^{(s)}_{n,m}$, the terms $T^{(sp)}_{n,m}$, $T^{(sg)}_{n,m}$, and $T^{(sv)}_{n,m}$ are related to $\gamma_{n,m}$. Since $T^{(sv)}_{n,m}$ is generally much smaller than $T^{(sp)}_{n,m}$ and $T^{(sg)}_{n,m}$, we only focus on $T^{(sp)}_{n,m}$ and $T^{(sg)}_{n,m}$ here. It's easy to know that when $\gamma_{n,m} = \frac{1}{1+\omega_b}$, $T^{(sp)}_{n,m} + T^{(sg)}_{n,m}$ takes the minimum value according to basic inequality. Following are the detailed steps:
\begin{talign}
    T^{(sp)}_{n,m} + T^{(sg)}_{n,m} \geq 2\sqrt{T^{(sp)}_{n,m} T^{(sg)}_{n,m}},
\end{talign}
where if and only if $T^{(sp)}_{n,m} =T^{(sg)}_{n,m}$, ``$=$'' can be obtained.
\begin{talign}
    &\quad \quad T^{(sp)}_{n,m} =T^{(sg)}_{n,m},\nonumber \\
    &\Rightarrow \frac{x_{n,m}\varphi_nd_n\eta_m}{\gamma_{n,m}\zeta_{n,m}f_m} = \frac{x_{n,m}\varphi_nd_n\omega_b\eta_m}{(1-\gamma_{n,m})\zeta_{n,m}f_m}, \nonumber \\
    &\Rightarrow \gamma_{n,m} = \frac{1}{1+\omega_b}.
\end{talign}
We set $\omega_b = 1$ and then $\frac{\omega_b}{1+\omega_b} = \frac{1}{1+\omega_b}$, in which case, $cost_{n,m}^{(s)}$ would take the minimum leading to Problem $\mathbb{P}_{7}$ take the maximum value. Based on the above discussion, we can obtain the optimal value of $\gamma_{n,m}$ that $\gamma_{n,m}^\star = \frac{\omega_b}{1+\omega_b}$ or $\frac{1}{1+\omega_b} = \frac{1}{2}$. 

Thus, \textbf{Lemma \ref{lemma_gamma}} is proven.
\end{proof}

\section{Proof of \textbf{Lemma \ref{lemma_p8top9}}}\label{append_lemma_p8top9}
\begin{proof}
Let $\bm{P}_0:=\bm{I}_{NM+N\times N} \bm{I}_{N\rightarrow NM} \text{diag}(\bm{B})\bm{e}_{N+1,NM+N}$ and we can represent $\sum_{n \in \mathcal{N}}\sum_{m \in \mathcal{M}}B_{n,m} x_{n,m} \varphi_n$ as $\bm{Q}^\intercal \bm{P}_0 \bm{Q}$. Let $\bm{W}_0^\intercal:=\bm{A}^\intercal\bm{e}_{1,N}$ and the term $\sum_{n \in \mathcal{N}}A_n \varphi_n$ can be rewritten as $\bm{W}_0^\intercal \bm{Q}$.  Let $P_{0,n}^{(T_u)}:= - \frac{\alpha^{(u)}_n \vartheta_n^{(u)} \omega_t d_n\eta_n}{\psi_nf_n}$, $P_1^{(T_u)}:=\sum_{n \in \mathcal{N}}\frac{\alpha^{(u)}_n \vartheta_n^{(u)} \omega_t d_n\eta_n}{\psi_nf_n}$, $\bm{P}_0^{(T_u)}:=[P_{0,n}^{(T_u)}]|_{n \in \mathcal{N}}$, and ${\bm{P}_2^{(T_u)}}^\intercal:={\bm{P}_0^{(T_u)}}^\intercal\bm{e}_{1,N}$. Then, the constraint (\ref{Tu_constr1}) can be represented by 
\begin{talign}
    {\bm{P}_2^{(T_u)}}^\intercal \bm{Q} + P_1^{(T_u)} \leq T^{(u)}.
\end{talign}
Let 
\begin{talign}
&P_{0,n,m}^{(T_s)}:=\frac{d_n}{r_{n,m}} + \frac{d_n\eta_m}{\gamma_{n,m}\zeta_{n,m}f_m} + \frac{d_n\omega_b\eta_m}{(1-\gamma_{n,m})\zeta_{n,m}f_m},\\
&\bm{P}_0^{(T_s)}:=[P_{0,n,m}^{(T_s)}]|_{n \in \mathcal{N},m \in \mathcal{M}},\\
&P_1^{(T_s)}:=\sum_{n \in \mathcal{N}}\sum_{m \in \mathcal{M}}\frac{S_b}{R_m} + \text{max}_{m^\prime\in\mathcal{M}\setminus\{m\}}\frac{\eta_v}{(1-\gamma_{n,m^\prime})\zeta_{n,m^\prime}f_m^\prime}.
\end{talign}
Similar to $\bm{P}_0$, let 
\begin{talign}
\bm{P}_2^{(T_s)}:= \bm{I}_{NM+N\times N} \bm{I}_{N\rightarrow NM} \text{diag}(\bm{P}_0^{(T_s)})\bm{e}_{N+1,NM+N}.
\end{talign}
Then, the constraint (\ref{Ts_constr1}) can be transformed into
\begin{talign}
    \bm{Q}^\intercal \bm{P}_0^{(T_s)} \bm{Q} + P_1^{(T_s)} \leq T^{(s)}.
\end{talign}
Let $\bm{\phi}:=(\phi_{1,1},\cdots,\phi_{N,M})^\intercal$ and $\bm{\zeta}:=(\zeta_{1,1},\cdots,\zeta_{N,M})^\intercal$. The constraints \text{(\ref{x_constr1_qcqp})}-\text{(\ref{x_zeta_constr_qcqp})} are easy to obtain, which we won't go into details here.

Therefore, \textbf{Lemma \ref{lemma_p8top9}} holds.
\end{proof}

\section{Proof of \textbf{Lemma \ref{lemma_p9top10}}}\label{append_lemma_p9top10}
\begin{proof}
Here we give the expression of matrices $\bm{P}_1$, $\bm{P}_2$, $\bm{P}_3$, $\bm{P}_4$, $\bm{P}_5$, $\bm{P}_6$, $\bm{P}_7$, and $\bm{P}_8$.
\begin{equation}
\bm{P}_1=
\left(
    \begin{array}{cc}
       \bm{P}_0  & \frac{1}{2}\bm{W}_0 \\
        \frac{1}{2}\bm{W}_0^\intercal &  T^{(u)} + T^{(s)} + C
    \end{array}
\right),
\end{equation}
\begin{equation}
\bm{P}_2=
\left(
    \begin{array}{cc}
      \boldsymbol{e}_{i}^\intercal\boldsymbol{e}_{i}   & -\frac{1}{2}\boldsymbol{e}_{i} \\
       -\frac{1}{2}\boldsymbol{e}_{i}^\intercal  & 0
    \end{array}
\right), \forall i \in \{1,\cdots, NM\}
\end{equation}
\begin{align}
\bm{P}_3=
\left(
    \begin{array}{cc}
    \bm{0}_{NM+N \times NM+N}     & \frac{1}{2}(\boldsymbol{e}_{\overline{1},\overline{M}}\boldsymbol{e}_{N+1,NM+N}^\intercal) \\
     \frac{1}{2}(\boldsymbol{e}_{\overline{1},\overline{M}}\boldsymbol{e}_{N+1,NM+N}^\intercal)^\intercal    & -1
    \end{array}
\right)\nonumber, \\ \forall i \in \{1,\cdots, N\}
\end{align}
\begin{equation}
\bm{P}_4=
\left(
    \begin{array}{cc}
      \bm{0}_{NM+N \times NM+N}   & \frac{1}{2}\boldsymbol{e}_{i} \\
       \frac{1}{2}\boldsymbol{e}_{i}^\intercal  & -1
    \end{array}
\right), \forall i \in \{1,\cdots, N\}
\end{equation}
\begin{equation}
\bm{P}_5=
\left(
    \begin{array}{cc}
    \bm{0}_{NM+N \times NM+N}    & \frac{1}{2}\boldsymbol{\phi}\boldsymbol{e}_{N+1,NM+N} \\
    \frac{1}{2}(\boldsymbol{\phi}\boldsymbol{e}_{N+1,NM+N})^\intercal     & -1
    \end{array}
\right),
\end{equation}
\begin{equation}
\bm{P}_6=
\left(
    \begin{array}{cc}
    \bm{0}_{NM+N \times NM+N}     & \frac{1}{2}\boldsymbol{\zeta}\boldsymbol{e}_{N+1,NM+N} \\
    \frac{1}{2}(\boldsymbol{\zeta}\boldsymbol{e}_{N+1,NM+N})^\intercal     & -1
    \end{array}
\right), 
\end{equation}
\begin{equation}
\bm{P}_7=
\left(
    \begin{array}{cc}
    \bm{0}_{NM+N \times NM+N}     & \frac{1}{2}\bm{P}_2^{(T_u)} \\
    \frac{1}{2}{\bm{P}_2^{(T_u)}}^\intercal     & P_1^{(T_u)}
    \end{array}
\right), 
\end{equation}
\begin{equation}
\bm{P}_8=
\left(
    \begin{array}{cc}
    \bm{P}_0^{(T_s)}     & \bm{0}_{NM+N \times 1}\\
    \bm{0}_{1 \times NM+N}     & P_1^{(T_s)}
    \end{array}
\right). 
\end{equation}

\textbf{Lemma \ref{lemma_p9top10}} is proven.
\end{proof}

\end{appendices}

\end{document}